\documentclass[reqno]{amsart}
\usepackage{amsmath}
\usepackage{amssymb}
\usepackage{amsthm}
\usepackage{amsfonts}
\usepackage{color}
\theoremstyle{plain} 
 
\newtheorem{thm}{Theorem}[section] 
 
\newtheorem{lemma}[thm]{Lemma} 
\newtheorem{prop}[thm]{Proposition}

\newtheorem{proposition}[thm]{Proposition}

\def\duh{{\rm Duh}}

\def\gammaN{\gamma_N}

\def\GammaN{\Gamma_N}

\def\opDelta{\widehat{\Delta}}
\def\opB{B}
\def\opBN{\opB_N}

\def\VN{V_N}

\def\tr{{\rm Tr}}

\def\bra{\big\langle}
\def\ket{\big\rangle}

\def\la{\langle}
\def\ra{\rangle}

\def\N{{\mathbb N}}

\def\R{{\mathbb R}}

\def\ux{{\underline{x}}}

\def\frH{{\mathfrak H}}

\def\cH{{\mathcal H}}

\def\1{{\bf 1}}

\def\eqnn{\begin{eqnarray*}}
\def\eeqnn{\end{eqnarray*}}
\def\eqn{\begin{eqnarray}}
\def\eeqn{\end{eqnarray}}

\def\prf{\begin{proof}}
\def\endprf{\end{proof}}

\def\fvar{u}

\numberwithin{equation}{section}

\title[Derivation and GWP of the GP Hierarchy]
{Derivation in strong topology and global well-posedness of solutions to the Gross-Pitaevskii hierarchy}
\author{Thomas Chen}
\address{T. Chen,  
Department of Mathematics, University of Texas at Austin.}
\email{tc@math.utexas.edu}
\author{Kenneth Taliaferro}
\address{K. Taliaferro,  
Department of Mathematics, University of Texas at Austin.}
\email{ktaliaferro@math.utexas.edu}
\begin{document}

\begin{abstract}
We derive the cubic defocusing GP hierarchy in $\R^3$ from a bosonic $N$-particle Schr\"odinger equation  
as $N\rightarrow\infty$, in the strong topology corresponding to the  
space $\cH_\xi^1$ introduced  in \cite{chpa}.
In particular, we thereby eliminate the requirement of regularity $\cH_\xi^{1+}$  for the initial data used in \cite{CPBBGKY}. 
Moreover, the marginal density matrices obtained in this strong limit are allowed to be of
infinite rank. This contrasts previous results where weak-* limits were derived,
and subsequently enhanced to strong limits
based on the condition that the limiting density matrices have finite rank.
Furthermore, we prove that positive semidefiniteness of marginal density matrices is 
preserved over time, which we combine with results in \cite{CPHE}, 
to obtain the global well-posedness of solutions. 
\end{abstract}
\maketitle

\section{Introduction}
\label{subsec-exsol-GP-0}

The Gross-Pitaevskii (GP) hierarchy emerges, in the limit as $N\rightarrow\infty$,
from an $N$-body Schr\"odinger equation describing an
interacting Bose gas under Gross-Pitaevskii scaling,
via the associated BBGKY (Bogoliubov–-Born–-Green–-Kirkwood–-Yvon) hierarchy.
Factorized solutions to the GP hierarchy are determined by a nonlinear Schr\"odinger equation (NLS).
Through this procedure, one obtains a rigorous derivation of the NLS as a mean field description of 
the dynamics of a Bose-Einstein condensate.   

In this paper, we extend previous results proven in \cite{CPBBGKY,CPHE}
concerning the derivation and global well-posedness of
the cubic GP hierarchy in $\R^3$.
We derive the GP hierarchy from the BBGKY hierarchy as $N\rightarrow\infty$ 
in the strong topology relative to the  space $\cH_\xi^1$ defined in
\eqref{eq-cHxi-def-0-1}, below, 
and we remove the requirement of regularity $\cH_\xi^{1+\delta}$, with an arbitrary $\delta>0$, for the initial data used in \cite{CPBBGKY}.
Moreover, we prove that solutions to the cubic defocusing GP hierarchy 
remain positive semidefinite over time, which we use, in combination with   
the higher order energy functionals and related results from \cite{CPHE}, to prove global well-posedness.

The first derivation of the nonlinear Hartree equation (NLH) as a mean field description of a
quantum manybody theory was given by Hepp in \cite{he1} using the Fock space formalism and coherent states.
Subsequently, Spohn provided a derivation of the NLH using the BBGKY hierarchy  in \cite{sp}.
More recently, this topic was revisited by Fr\"ohlich, Tsai and Yau  in \cite{frtsya}.
In a series of very important works, Erd\"os, Schlein and Yau gave the derivation of 
the NLS and NLH for a wide range of situations  \cite{esy1,esy2,esy3,esy4};
we will outline the main steps of their construction below, and will also mention related works of other authors.

The problems studied here are closely related to the study of Bose-Einstein condensation, 
where fundamental progress was made
in recent years, see \cite{ailisesoyn,lise,lisesoyn,liseyn, liseyn2} and the references therein.

\subsection{Definition of the Model and Background} 
   
We consider a system of $N$ bosons in $\R^3$ described by a wave function 
$\Phi_{N} \in L_{sym}^2(\R^{3N})$ that satisfies the $N$-body Schr\"odinger equation
\begin{align}\label{ham1}
    i\partial_t\Phi_N=H_N\Phi_N, 
\end{align}
where the Hamiltonian $H_N$ is the self-adjoint operator on $L^2(\mathbb{R}^{3N})$ given by
\begin{align}\label{ham2}
    H_{N}=\sum_{j=1}^{N}(-\Delta_{x_{j}})+\frac1N\sum_{1\leq i<j \leq N}V_N(x_{i}-x_{j}).
\end{align}
Here, $L^2_{Sym}(\mathbb{R}^{3N})$ is the subspace of $L^2(\mathbb{R}^{3N})$ that is invariant under permutations \eqref{sym} of the $N$ particle variables,  
\begin{align}\label{sym}
	\Phi_N(x_{\pi( 1)},x_{\pi (2)},...,x_{\pi (N)})=\Phi_N(x_1, x_2,..., x_N) 
	\; \; \; \; \; \; \; \; \forall \pi \in S_{N} \,,
\end{align} 
where  $S_{N}$ is the $N$-th symmetric group.  
We note that permutation symmetry \eqref{sym} is preserved by the $N$-body Schr\"odinger equation \eqref{ham1}.
The potential $V_N$ satisfies 
\eqn\label{betadef}
    V_N(x)=N^{d\beta}V(N^\beta x)\,,
\eeqn
where $V \in \mathcal{S}(\R^3)\backslash\{0\}$ is spherically symmetric and
nonnegative. 
The parameter $\beta$ typically has values in $(0,1]$ (see the discussion below equation \eqref{beta});
in this work, we will have $0<\beta<\frac14$. 

Since the $N$-body Schr\"{o}dinger equation \eqref{ham1} is linear and $H_N$ is self-adjoint,  the global 
well-posedness of solutions in $L^2_{sym}(\mathbb{R}^{3N})$ is clear.

\subsubsection{BBGKY hierarchy}
To begin with, one considers the density matrix 
\eqn
    \gamma_{\Phi_N}^{(N)}=|\Phi_N\rangle\langle\Phi_N| \,,
\eeqn
and its marginals, 
\begin{align}\label{eq-gammaPhiNk-def-1}
 \gamma_{\Phi_N}^{(k)}(t,\underline{x}_k,\underline{x}_k'):=
 \int\Phi_N(t,\underline{x}_k,\underline{x}_{N-k}) 
 \overline{\Phi_N(t,\underline{x}_k',\underline{x}_{N-k})}\,d\underline{x}_{N-k}
 \;,\;\;1\le k\le N\,,
\end{align} 
where $(\underline{x}_k,\underline{x}_{N-k})\in\mathbb{R}^{3k}\times\mathbb{R}^{3(N-k)}$.
The marginal density matrices satisfy the property of {\em admissibility},
\begin{align}\label{eq-admissdef-1}
	\gamma^{(k)}_{\Phi_N}=\tr_{k+1}(\gamma^{(k+1)}_{\Phi_N})
	, \; \; \; \; k= 1,\dots,N-1 \,,
\end{align}
they define positive semidefinite operators on $L^2(\mathbb{R}^{3k})$,
and  $\tr \gamma_{\Phi_N}^{(k)}=\|\Phi_N\|_{L_{s}^{2}(\mathbb R^{3N})}^2=1$, for an arbitrary $N\in\N$. 

Moreover, each $\gamma_{\Phi_N}^{(k)}$ is completely symmetric under permutation of particle variables,
and hermitean. That is,   
\begin{align}\label{symmetry}
	\gamma^{(k)}(x_{\pi (1)}, ...,x_{\pi (k)};x_{\pi'( 1)}^{\prime},
	 ...,x_{\pi'(k)}^{\prime})&=\gamma^{(k)}( 	x_1, ...,x_{k};x_{1}^{\prime}, ...,x_{k}^{\prime})\qquad\text{and}\\ \gamma^{(k)}(\ux_k;\ux_k')&=\overline{\gamma^{(k)}(\ux_k';\ux_k) }\nonumber
\end{align}
for all permutations $\pi,\pi'\in S_k$.

It follows from the $N$-body Schr\"odinger equation \eqref{ham1} that  
\begin{align}\label{von}
	i\partial_{t}\gamma_{\Phi_N}^{(N)}(t)=[H_{N}, \gamma_{\Phi_N}^{(N)}(t)] \, .
\end{align} 
Accordingly, the $k$-particle marginals satisfy the BBGKY hierarchy 
\begin{align} 
	i\partial_{t}\gamma_{\Phi_N}^{(k)}(t,\ux_k;\ux_k')
	 &=
	- \; (\Delta_{\ux_k}-\Delta_{ \ux_k'})\gamma_{\Phi_N}^{(k)}(t,\ux_k,\ux_k')
	\nonumber\\
	&\hspace{0.5cm}
	+ \frac{1}{N}\sum_{1\leq i<j \leq k}[V_N(x_i-x_j)-V_N(x_i^{\prime}-x_{j}')]
	\gamma_{\Phi_N}^{(k)}(t, \ux_{k};\ux_{k}') 
	 \nonumber\\
	&\hspace{0.5cm}+\frac{N-k}{N}\sum_{i=1}^{k}\int dx_{k+1}[V_N(x_i-x_{k+1})-V_N(x_i^{\prime}-x_{k+1})]
	\label{eq-bbgky-2}\\	&\hspace{5cm}
	\gamma_{\Phi_N}^{(k+1)}(t, \ux_{k},x_{k+1};\ux_{k}',x_{k+1})
	\label{BBGKY-0}
\end{align}
for $1\leq k < N$,
where $\Delta_{\ux_k}:=\sum_{j=1}^{k}\Delta_{x_j}$, and similarly for $\Delta_{\ux_k'}$.  
 
\subsubsection{Derivation of the GP hierarchy.} 
In \cite{esy1,esy2,ey}, the authors consider 
factorizing initial data, i.e.,
\begin{align}
    \gamma_{\Phi_N}^{(k)}(0,\ux_k;\ux_k')
    \,=\,
     \gamma_0^{(k)}(\ux_k;\ux_k')=\prod_{j=1}^{k}\phi_0( x_{j})\overline{\phi_0( x_{j}^{\prime}})\,,\label{fac}
\end{align}
as $N\rightarrow\infty$, where $\phi_0\in H^1(\mathbb{R}^{3k})$.
In particular, they prove that in the limit $N\rightarrow\infty$, solutions to 
the BBGKY hierarchy converge
in the weak-* topology, $\gamma^{(k)}_{\Phi_{N}}\rightharpoonup^*\gamma^{(k)}$ for $k\in\N$, 
on the space of trace class marginal
density matrices. 
Moreover, it is proven in \cite{esy1,esy2,ey} that the marginal density 
matrices obtained in the weak-* limit
satisfy the infinite hierarchy  
\begin{align}
	i\partial_{t}\gamma^{(k)}(t,\ux_k;\ux_k')
	&=
 	- \, (\Delta_{\ux_k}-\Delta_{\ux_k'})\gamma^{(k)}(t,\ux_k;\ux_k')
	\label{eq-GP-0-0}\\ 
	&\hspace{0.5cm}+ \, \kappa_0 \sum_{j=1}^{k}   \left(B_{j, k+1} 
	\gamma^{(k+1)}\right)(t,\ux_k  ; \ux_k' )  \;\;,\;\;\;k\in\N\,,
	\label{eq-GPdef-1-0}
\end{align}
which is referred to as the {\em Gross-Pitaevskii (GP) hierarchy}.
The interaction operator is defined by
\begin{align}
	&(B_{j, k+1} \gamma^{(k+1)})(t,  \ux_k ; \ux_k' )
	\nonumber\\
	&:=
	\int dx_{k+1}dx_{k+1}'[\delta(x_j-x_{k+1})\delta(x_{j}-x_{k+1}^{\prime})-\delta(x_j^{\prime}-x_{k+1})
	\delta(x_{j}^{\prime}-x_{k+1}^{\prime})]
	\nonumber\\	&\quad\quad\quad\quad\quad\quad\quad\quad\quad\quad\quad\quad\quad\quad\quad\quad
	\gamma^{(k+1)}(t,\ux_k, x_{k+1};\ux_k', x'_{k+1}) \,.
	\label{beta} 
\end{align}
In the case $\beta<1$, one obtains a coupling constant $\kappa_0=\int V(x) dx$ 
(corresponding to the Born approximation of the scattering length).
In the case $\beta=1$, the coefficient $\kappa_0$ is the {\em scattering length}; 
the derivation of the GP hierarchy in this case is much more difficult, \cite{esy1,esy2}.
We will have $0<\beta<\frac14$ in this paper, and set $\kappa_0=1$.

In \cite{esy2,esy1,ey}, solutions of the GP hierarchy are studied in   spaces
of $k$-particle marginals 
\eqn\label{eq-frH1-def-1}
    \frH^1\,:=\,\{(\gamma^{(k)})_{k\in\N} \, | \,\tr (|S^{(k,1)}\gamma^{(k)}|) \, < \, M^k
    \;{\rm for\;some\;constant\;}0<M<\infty\}
\eeqn  
where for $\alpha>0$,
\begin{align}
	S^{(k,\alpha)}:= 
	\prod_{j=1}^k (1-\Delta_{x_j})^{\alpha/2}(1-\Delta_{x_j'})^{\alpha/2} \,.
\end{align}
The solutions to the GP hierarchy obtained from the weak-* limit as described above exist  
{\em globally} in $t$, and are positive semidefinite.

\subsubsection{NLS from factorized solutions of GP}

Given factorized initial data \eqref{fac}, one can easily verify that
\begin{align*}
\gamma^{(k)}(t,\ux_k;\ux_k')=\prod_{j=1}^{k}\phi(t,x_{j})\overline{\phi(t,x_{j}^{\prime})}
\end{align*}
is a solution  (referred to as a {\em factorized solution})
of the  GP hierarchy \eqref{eq-GP-0-0}   if $\phi(t)\in H^1(\R^d)$ solves the defocusing cubic NLS,
\begin{align}
	i\partial_t\phi=- \Delta_x \phi + |\phi|^2\phi\,,
\end{align}
for $t\in I\subseteq\R$, and $\phi(0)=\phi_0\in H^1(\R^d)$.  
In this precise sense, the NLS emerges as a mean field 
description of the dynamics of Bose-Einstein condensates.

\subsubsection{Uniqueness of solutions of GP hierarchies.}

The most involved part in this analysis is the proof of
uniqueness of solutions to the GP hierarchy,
and was achieved in Erd\"os-Schlein-Yau in \cite{esy2,esy1,ey}, 
in the space $\frH^1$ using  
high dimensional singular integral estimates organized with 
Feynman graph expansions.

Subsequently, Klainerman and Machedon \cite{KM} presented 
an alternative method for proving uniqueness in a  space of density matrices defined by the 
Hilbert-Schmidt type Sobolev norms 
\begin{align}\label{eq-gamma-norm-def-0-1}
	\| \gamma^{(k)} \|_{H^1_k}:=\| S^{(k,1)} \gamma^{(k)} \|_{L^2(\R^{3k} \times \R^{3k})} \,.
\end{align}
While this is a different (strictly larger) space of marginal density matrices than
the one considered by Erd\"os, Schlein, and Yau, \cite{esy1,esy2},
the authors of \cite{KM} impose an additional a priori condition on  
space-time norms of the form
\begin{align} \label{intro-KMbound} 
     \|B_{j;k+1} \gamma^{(k+1)}\|_{L^2_{t\in[0,T]} H^1_k}<C^k \, , 
\end{align}
for some arbitrary but finite $C$ independent of $k$. 
The method of Klainerman-Machedon in \cite{KM} combines techniques from dispersive nonlinear PDE's
with a reformulation of the combinatorial method introduced in \cite{esy1,esy2,esy3,esy4}, which is referred to 
as a ``boardgame'' argument.

The Klainerman-Machedon framework was used by  Kirkpatrick, Schlein, and Staffilani  in  \cite{kiscst}, to give a derivation of the cubic defocusing NLS in dimensions $d=1,2$, and by Chen and Pavlovi\'c in  \cite{chpa}, to derive the quintic NLS for $d=1,2$.

As another line of research in this area, 
the study of the well-posedness theory of the GP hierarchy was initiated in 
\cite{CP,chpatz1,chpa4,CPHE}.
The authors introduced Banach spaces of sequences of marginal density matrices 
$\Gamma\in\bigoplus_{k=1}^\infty L^2(\R^{3k}\times\R^{3k})$
\eqn\label{eq-cHxi-def-0-1}
    \cH_\xi^\alpha \, := \, \{ {\rm symmetric} \; \Gamma=(\gamma^{(k)})_{k\in\N} \,| \,
    \| \, \Gamma \, \|_{\cH_\xi^\alpha} <\infty \}
\eeqn  
with 
\begin{align}\label{eq-KM-aprioriassumpt-0-1}
	\| \, \Gamma \, \|_{\cH_\xi^\alpha}:=	\sum_{k\in\N} \xi^k \, \| \, \gamma^{(k)} \, \|_{H^\alpha} 
	\;\;,\;\;\;\xi>0\,,
\end{align}
where
\begin{align}\label{eq-gamma-norm-def-0-0}
	\| \gamma^{(k)} \|_{H^\alpha}:=
	\| S^{(k,\alpha)} \gamma^{(k)} \|_{L^2(\mathbb{R}^{3k}\times\mathbb{R}^{3k})}
\end{align} 
is of Hilbert-Schmidt type, as in \eqref{eq-gamma-norm-def-0-1}.   
Those spaces are equivalent to those considered by Klainerman and Machedon in \cite{KM}.
Here, we call $\Gamma =(\gamma^{(k)})_{k\in\N}$ {\em symmetric} if each $\gamma^{(k)}(\ux_k,\ux_k')$
satisfies \eqref{symmetry}.
Moreover, we call $\Gamma$ {\em positive semidefinite} if $\gamma^{(k)}$ defines a positive semidefinite integral
operator on $L^2(\mathbb{R}^{3k})$ for all $k$.

We also define the spaces
\begin{align}\label{frakH-def-1}
 	\frH_\xi^\alpha:=\Big\{ \,\text{symmetric }\Gamma \, \in \, \bigoplus_{k=1}^\infty L^2(\R^{3k}\times\R^{3k}) \, \Big| \, 
 	\|\Gamma\|_{\frH_\xi^\alpha} < \, \infty \, \Big\}
\end{align}
where for $\alpha>0$,
\begin{align}\nonumber 
	\|\Gamma\|_{\frH_\xi^\alpha}=\sum_{k=1}^\infty \xi^{ k} 
	\tr(| S^{(k,\alpha)}  \gamma^{(k)} |) \,
\end{align}
is of trace-norm type.
Clearly, $\frH^1=\cup_{\xi>0}\frH^1_\xi$ corresponding to \eqref{eq-frH1-def-1}.
Moreover, $\frH_\xi^\alpha\subset\cH_\xi^\alpha$ holds for any $\alpha>0$ and $0<\xi<1$.
 
In \cite{CPBBGKY}, Chen and Pavlovi\'c proved that solutions of the $N$-body Schr\"odinger equation 
converge to the solution of the GP hierarchy in $\cH_\xi^1$, 
for values $\beta\in(0,1/4)$, provided that the initial data satisfies $\Gamma_0\in\frH_{\xi'}^{1+\delta}$ for an arbitrary $\delta>0$,
and for $\xi$ sufficiently small (depending on $0<\xi'<1$).  
No factorization of solutions was assumed.

More recently, 
a derivation of the GP hierarchy in Klainerman-Machedon type spaces 
was given by X. Chen and J. Holmer in \cite{ChXHol13}, for
values $\beta\in(0,2/3)$. Assuming a regularity requirement that follows from the condition \eqref{energy condition} in Theorem \ref{derivation} of our paper, 
they prove that solutions to the BBGKY hierarchy converge to solutions of the GP hierarchy
satisfying the Klainerman-Machedon condition \eqref{intro-KMbound}.
This convergence is shown in the weak-* topology on the space of trace class marginal density matrices.  
See also  \cite{xch3}.

Rodnianski and Schlein in \cite{rosc} investigated the rate of convergence to the NLH, based on the approach 
of Hepp \cite{he1}, which led to many further developments, including works of 
Grillakis-Machedon and G-M-Margetis  \cite{grma2012,grma,grmama2,grmama},
X.Chen \cite{xch1}, and Lee-Li-Schlein \cite{LiLeeSchlein11}. 
Many authors have contributed to this very active research field, and introduced a variety of different
approaches; see for instance \cite{adgote,anasig,eesy,frgrsc,frknpi,frknsc,pick}.  

\subsection{Outline of Main Results}

The main results of this paper can be summarized as follows: 
\begin{itemize}
\item
\underline{\em Derivation in $\cH^1_\xi$:}
We show that solutions to the $N$-BBGKY hierarchy with initial data $\Gamma_{0,N}$ converge to those of the GP hierarchy strongly in $C([0,T],\mathcal{H}^1_\xi)$ as $N\rightarrow\infty$ 
when the intial data is in $\mathfrak{H}^1_\xi$, see \eqref{frakH-def-1}.
In \cite{CPBBGKY}, this convergence is obtained with initial data in $\mathcal{H}^{1+\delta}_{\xi'}$, 
for an arbitrary, small $\delta>0$ extra regularity.
In this paper, we eliminate this condition, and provide the derivation of the 
GP hierarchy in the energy space.  
The detailed discussion is given in Section \ref{positivity_proof_section}.
\\
\item
\underline{\em Strong convergence for limits of infinite rank:}
The convergence proven in the work at hand  is established in the 
{\em strong} topology on $\mathcal{H}^1_\xi$.  

We note that in previous works following the BBGKY approach  (except for \cite{CPBBGKY}), including \cite{esy1,esy2,esy3,esy4,ChXHol13}, 
convergence along a subsequence
is shown in the {\em weak-*} topology on the space of trace class marginal density matrices,
using a compactness argument.  Subsequently, this weak-* convergence is enhanced to
strong convergence in trace norm, due to the special case of 
the limiting density matrices being of {\em finite rank} (since factorized initial data are considered).
In the finite rank case
(i.e., the rank of $\gamma^{(k)}$ is bounded uniformly in $k$), the trace norm, corresponding to  $\frH_\xi^1$, 
is equivalent to the Hilbert-Schmidt norm, corresponding to $\cH_\xi^1$.
In contrast, 
we obtain a strong limit in $\cH_\xi^1$ without any finite rank requirement.
\\
\item
\underline{\em Global well-posedness:}
Combining the higher order energy functional introduced in  \cite{CPHE}, combined with the quantum de Finetti theorem
as formulated by Lewin-Nam-Rougerie in \cite{lenaro}, we prove that solutions to the cubic defocusing GP hierarchy
are globally well-posed. To this end, we prove that those solutions remain positive semidefinite over time if 
the initial data are positive semidefinite; this allows us to invoke a key result in \cite{CPHE} to arrive at global well-posedness.
This is carried out in Section \ref{sec-gwp-1}.
\\
\item
\underline{\em Global in time derivation of GP hierarchy:}
By combining our derivation of the GP hierarchy locally in time, and global well-posedness of the GP hierarchy, 
we arrive at a derivation of the GP hierarchy on arbitrarily large time intervals $[0,T]$; 
the details are presented in Section \ref{global derivation}.   
\\
\end{itemize}

\section{Statement of Main Theorems}
  
In this section, we present the main theorems proven in this paper. Our first main result provides the derivation of the GP hierarchy
from a bosonic $N$-body Schr\"odinger system via the associated BBGKY hierarchy as $N\rightarrow\infty$,
in the {\em  strong topology} with respect to the energy space $L^\infty_{t\in[0,T]}\cH_\xi^1$, for a suitable $0<\xi<1$.
In particular, this strong limit yields solutions $\Gamma(t)=(\gamma^{(k)}(t))_k$ to the GP hierarchy where
$\gamma^{(k)}(t)$ does not need to be of finite rank. As noted above, a strong limit was obtained in the previous works
\cite{esy1,esy2,esy3,esy4,ChXHol13} only for the special case where the 
density matrices $(\gamma^{(k)}(t))_k$  have finite rank.

Moreover, our result removes an extra regularity condition on the initial data which was
assumed in \cite{CPBBGKY}. In \cite{CPBBGKY}, solutions to the GP hierarchy were derived from the
BBGKY hierarchy under
the requirement that $\Gamma_0\in\mathfrak{H}^{1+\delta}_{\xi'}$
for an arbitrarily small, but positive $\delta>0$. Here, we assume that $\Gamma_0\in\mathfrak{H}^{1}_{\xi'}$.

\begin{thm}
\label{derivation}
Let $(\Phi_N)_N$ be a sequence of solutions to the N-body Schr\"odinger equation \eqref{ham1}
with the corresponding marginal density matrices 
$\gamma_{\Phi_N}^{(k)}(t)$ given by \eqref{eq-gammaPhiNk-def-1}.
Suppose that
\begin{align}
    \la\Phi_N(0),H_N^k\Phi_N(0)\ra<C^kN^k
    \label{energy condition}
\end{align}
and $\|\Phi_N(0)\|_{L^2}=1$ for all $N\in\mathbb{N}$ and $k\le N$, where $C$ does
not depend on $k$ or $N$.  
Moreover, assume that for some $0<\xi'<1$, and every $N\in\mathbb{N}$, we have
\begin{align*}
    \Gamma^{\Phi_N}(0)=(\gamma_{\Phi_N}^{(1)}(0),\dots,\gamma_{\Phi_N}^{(N)}(0),0,\cdots)\in\cH_{\xi'}^{1}
\end{align*}
and that
\begin{align*}
    \Gamma_0:=\lim_{N\rightarrow\infty}\Gamma^{\Phi_N}(0)
\end{align*}
exists in $\cH_{\xi'}^1$.
Define the truncation operator $P_{\le K}$ by
\begin{align*}
    P_{\le K(N)}\Gamma=(\gamma^{(1)},\dots,\gamma^{(K(N))},0,\dots),
\end{align*}
where $\tfrac{1}{2}b_1\log N\le K(N)\le b_1\log N$ for some $b_1>0$.
Then, for sufficently small $b_1>0$ (depending only on $\beta$ (see \eqref{betadef})) and sufficiently small $\xi>0$ (depending on only on $\xi'$) and sufficiently small $T>0$ (depending only on $\xi$), the limit
\begin{align*}
    \Gamma:=\lim_{N\rightarrow\infty}P_{\le K(N)}\Gamma^{\Phi_N}
\end{align*}
exists in $L^\infty_{t\in [0,T]}\mathcal{H}^1_\xi$ and satisfies the GP hierarchy with initial data $\Gamma_0\in\frH^1_{\xi'}$.  Moreover,
\begin{align*}
    B\Gamma=\lim_{N\rightarrow\infty}B_NP_{K\le N}\Gamma^{\Phi_N}
\end{align*}
holds in $L^2_{t\in [0,T]}\mathcal{H}^1_\xi$.
The abbreviated notations $B$ and $B_N$ for the interaction operators
are defined in \eqref{chpa2-pGP} and \eqref{eq-BBGKY-condensed-B}, below.
\end{thm}

Our second main result establishes global well-posedness for solutions to the GP hierarchy, 
where our proof uses the quantum de Finetti theorem in the formulation presented 
in a
recent paper by Lewin, Nam and Rougerie \cite{lenaro}, which we quote here:

\begin{thm}\label{qdf} 
{\em (Quantum de Finetti Theorem)}
Let $\cH$ be a separable Hilbert space and let 
$\cH^k = \bigotimes_{sym}^k\cH$ denote the corresponding bosonic $k$-particle space. 
Let $\Gamma$ denote a collection of   
bosonic density matrices on  $\cH$, i.e.,
$$
\Gamma = (\gamma^{(1)},\gamma^{(2)},\dots)
$$
with $\gamma^{(k)}$ a non-negative trace class operator on $\cH^k$.
Then, the following hold:
\begin{itemize}
\item
(Strong Quantum de Finetti theorem, \cite{HudsonMoody,Stormer-69,lenaro})
Assume that $\Gamma$ is admissible, i.e., $\gamma^{(k)}=\tr_{k+1} \gamma^{(k+1)}$, 
where $\tr_{k+1}$ denotes the partial trace over the $(k+1)$-th factor, $\forall k\in\N$. 
Then, there exists a unique Borel probability measure $\mu$, 
supported on the unit sphere in $\cH$, and invariant under multiplication of 
$\phi \in \cH$ by complex numbers of modulus one, such that 
\begin{equation}\label{gkdf}
    \gamma^{(k)} = \int d\mu(\phi)  (  | \phi  \rangle \langle \phi |  )^{\otimes k}
    \;\;\;,\;\;\;\forall k\in\N\,.
\end{equation} 
\item 
(Weak Quantum de Finetti theorem, \cite{lenaro,AmmariNier-2008,AmmariNier-2011}) 
Assume that $\gamma_N^{(N)}$ is an arbitrary sequence of mixed states on $\cH^N$, $N\in\N$,
satisfying $\gamma_N^{(N)}\geq 0$ and $\tr_{\cH^N}(\gamma_N^{(N)})=1$, and assume
that its $k$-particle marginals have weak-* limits 
\eqn 
    \gamma^{(k)}_{N}:=\tr_{k+1,\cdots,N}(\gamma^{(N)}_N)
    \; \rightharpoonup^* \; \gamma^{(k)} \;\;\;\; (N\rightarrow\infty)\,,
\eeqn
in the trace class on $\cH^k$ for all $k\geq1$ (here, $\tr_{k+1,\cdots,N}(\gamma^{(N)}_N)$ 
denotes the partial trace in the $(k+1)$-st up to $N$-th component). 
Then, there exists a unique Borel probability measure $\mu$ on the unit ball in
$\cH$, and invariant under multiplication of 
$\phi \in \cH$ by complex numbers of modulus one, 
such that
\eqref{gkdf} holds
for all $k\geq0$. 
\end{itemize}
\end{thm}

In the context of Theorem \ref{derivation} proven here, 
strong convergence in Hilbert Schmidt norm implies weak-* convergence in the trace norm topology, 
provided that the limit point is trace class, see Proposition \ref{strong weak} in the Appendix.
Therefore, the solutions to the GP hierarchy obtained 
in Theorem \ref{derivation} from the BBGKY hierarchy \eqref{BBGKY-0},
in the limit as $N\rightarrow\infty$, satisfy the 
conditions of the quantum de Finetti theorem, either in its strong or its weak form.

Moreover, due to the bound \eqref{energy condition}, it follows that the initial data for the
GP hiearchy satisfies $\Gamma_0\in\frH_{\xi'}^1$ for some $0<\xi'<1$. This implies that
the higher order energy functionals  introduced in \cite{CPHE}, which correspond to 
$\Gamma_0$, are well-defined; see Section \ref{sec-gwp-1}.
We are therefore able to combine an application of 
the quantum de Finetti theorem with  
the higher order energy functionals for the cubic GP hierarchy that
were introduced in \cite{CPHE}, to prove the following result on the global well-posedness of solutions to the
defocusing cubic GP hierarchy.

\begin{thm}\label{gwp}
Assume that 
\eqn\label{eq-gamma0-deF-1}
    \gamma_0^{(k)} = \int d\mu(\phi)(|\phi\rangle\langle\phi|)^{\otimes k}
    \;,\;\;k\in\N\,
\eeqn
satisfies $\Gamma_0=(\gamma_0^{(k)})_{k=1}^\infty\in\mathfrak{H}^1_{\xi'}$ for some $0<\xi'<1$,
where $d\mu$ is a probability measure supported either on the unit sphere, or on the unit ball in $L^2(\R^3)$.
For $I\subseteq \mathbb{R}$, we denote by
\begin{align}\label{eq-Wxialph-def-1}
    \mathcal{W}^\alpha_\xi(I):=\{\Gamma\in C(I,\mathcal{H}^\alpha_\xi)\,|\, 
    B^+\Gamma,B^-\Gamma\in L^2_{loc}(I,\mathcal{H}^\alpha_\xi)\} \,,
\end{align} 
the space of local in time solutions of the GP hierarchy, with $t\in I$, following \cite{CP}.

Then, for a sufficiently small $\xi_1>0$ (depending only on $\xi'$), 
there is a unique global solution $\Gamma\in\mathcal{W}^1_{\xi_1}(\mathbb{R})$ to the cubic 
defocusing GP hierarchy \eqref{gp} in $\mathbb{R}^3$ with initial data $\Gamma_0$.  
Moreover, $\Gamma(t)$ is positive semidefinite and satisfies
\begin{align*}
\|\Gamma(t)\|_{\mathcal{H}_{\xi_1}^1}\le \|\Gamma_0\|_{\mathfrak{H}_{\xi'}^1}
\end{align*}
for all $t\in\mathbb{R}$.
The dependence of $\xi_1$ on $\xi'$ is detailed in Section \ref{ssec-constants-1}, below.
\end{thm}

\subsection{Remarks}
\begin{itemize}
\item
We note that, by combining Theorem \ref{derivation} and \ref{gwp}, one can show that Theorem \ref{derivation} actually holds for $T$ arbitrarily large, provided that $\Gamma_0\in\mathfrak{H}^1_{\xi'}$, 
and that $\xi$ is sufficiently small.  This is addressed in detail in Section \ref{global derivation}.
\item
Although we only address the the cubic defocusing GP hierarchy in $\mathbb{R}^d$ for $d=3$, we note that Theorem \ref{gwp} can be proved in the same way for the more general cases considered in Theorem 7.2 of \cite{CPHE}.  Let $\kappa_0$ be the constant in \eqref{eq-GPdef-1-0}, and let $p=2,4$ correspond to the cubic and quintic GP hierarchies, respectively.  Then, we have global well-posedness for the following cases:
\begin{itemize}
\item Energy subcritical, defocusing GP hierarchy with $p<\tfrac{4}{d-2}$ and $\kappa_0=+1$.
\item $L^2$ subcritical, focusing GP hierarchy with $p<\tfrac{4}{d}$ and $\kappa_0<0$ with $|\kappa_0|$ sufficiently small (see Theorem 7.2 in \cite{CPHE} for an explicit bound on $|\kappa_0|$).
\end{itemize}
\item
A solution to the GP hierarchy obtained in a weak-* limit does not necessarily satisfy admissibility,
even if the system at finite $N$ does.
However, we note that solutions to the GP hierarchy preserve   admissibility \eqref{eq-admissdef-1},
provided that it holds at the initial time $t=0$; see Proposition \ref{adm_prop} in the appendix.
\end{itemize}

\section{Notations for the GP and BBGKY Hierarchy} 
\label{sec-not}

For convenience, we introduce additional notations for the GP and the BBGKY hierarchy in this section,
mostly adopted from \cite{CP},
which allow us to discuss them both on the same setting.

Let $0<\xi<1$. We recall that
\begin{align}\label{eq-cHalpha-def-1} 
 	\cH_\xi^\alpha=\Big\{ \,\text{symmetric }\Gamma \, \in \, \bigoplus_{k=1}^\infty L^2(\R^{3k}\times\R^{3k}) 
 	\, \Big| \, \|\Gamma\|_{\cH_\xi^\alpha} < \, \infty \, \Big\}
\end{align}
where
\begin{align}\nonumber 
	\|\Gamma\|_{\cH_\xi^\alpha}=\sum_{k=1}^\infty \xi^{ k} 
	\| \,  \gamma^{(k)} \, \|_{H^\alpha(\R^{3k}\times\R^{3k})} \,,
\end{align}
with
\begin{align}\label{eq-gamma-norm-def-1}
	\| \gamma^{(k)} \|_{H^\alpha} & :=
	\| S^{(k,\alpha)}  \gamma^{(k)} \|_{L^2(\mathbb{R}^{3k}\times\mathbb{R}^{3k})}
\end{align}
where $S^{(k,\alpha)}=\prod_{j=1}^k(1-\Delta_{x_j})^{\alpha/2}(1-\Delta_{x_j'})^{\alpha/2}$.

We also recall the spaces
\begin{align}\label{frakH-def-1}
 	\frH_\xi^\alpha=\Big\{ \,\text{symmetric }\Gamma \, \in \, \bigoplus_{k=1}^\infty L^2(\R^{3k}\times\R^{3k}) \, \Big| \, \|\Gamma\|_{\frH_\xi^\alpha} < \, \infty \, \Big\}
\end{align}
where
\begin{align}\nonumber 
	\|\Gamma\|_{\frH_\xi^\alpha}=\sum_{k=1}^\infty \xi^{ k} 
	\tr(| S^{(k,\alpha)}  \gamma^{(k)} |) \,.
\end{align}
is of trace-norm type.


\subsection{The GP hierarchy}

The cubic defocusing GP hierarchy is given by
\begin{align}
	i\partial_t \gamma^{(k)}=\sum_{j=1}^k\, [-\Delta_{x_j},\gamma^{(k)}]   
	+ B_{k+1} \gamma^{(k+1)}\label{gp}
\end{align}
for $k\in\N$, where
\begin{align} \label{eq-def-b}
	B_{k+1}\gamma^{(k+1)}
	=B^+_{k+1}\gamma^{(k+1)}
        - B^-_{k+1}\gamma^{(k+1)} \, ,
\end{align}
where 
\begin{align}\label{eq-Bplus-GP-def-1}
	B^+_{k+1}\gamma^{(k+1)}
   = \sum_{j=1}^k B^+_{j;k+1 }\gamma^{(k+1)},
\end{align}
and 
\begin{align} 
	B^-_{k+1}\gamma^{(k+1)}
   = \sum_{j=1}^k B^-_{j;k+1 }\gamma^{(k+1)},
\end{align}     
with 
\begin{align*} 
& \left(B^+_{j;k+1}\gamma^{(k+1)}\right)
(t,x_1,\dots,x_k;x_1',\dots,x_k') \\
& \quad \quad = \int dx_{k+1}  dx_{k+1}'  \\
& \quad\quad\quad\quad 
	\delta(x_j-x_{k+1})\delta(x_j-x_{k+1}' )
        \gamma^{(k+1)}(t,x_1,\dots,x_{k+1};x_1',\dots,x_{k+1}'),
\end{align*} 
and 
\begin{align*} 
& \left(B^-_{j;k+1}\gamma^{(k+1)}\right)
(t,x_1,\dots,x_k;x_1',\dots,x_k') \\
& \quad \quad = \int dx_{k+1} dx_{k+1}'  \\
& \quad\quad\quad\quad 
	  \delta(x'_j-x_{k+1})\delta(x'_j-x_{k+1}' )
        \gamma^{(k+1)}(t,x_1,\dots,x_{k+1};x_1',\dots,x_{k+1}').
\end{align*}

The GP hierarchy can be rewritten in the following compact manner:
\begin{align} \label{chpa2-pGP}
        i\partial_t \Gamma + \opDelta_\pm \Gamma & =   \opB \Gamma 
        \nonumber\\
        \Gamma(0) &= \Gamma_0 \,,
\end{align}
where
$$
	\opDelta_\pm \Gamma := ( \, \Delta^{(k)}_\pm \gamma^{(k)} \, )_{k\in\N} \, ,
        \; \; \; \; \mbox{ with }
        \Delta_{\pm}^{(k)} = \sum_{j=1}^{k} \left( \Delta_{x_j} - \Delta_{x'_j} \right)\, ,
$$
and 
\begin{align} \label{chpa2-B} 
	\opB \Gamma:=( \, B_{k+1} \gamma^{(k+1)} \, )_{k\in\N} \,.
\end{align}
We will also use the notation 
\begin{align*} 
	\opB^+ \Gamma :=( \, B^+_{k+1} \gamma^{(k+1)} \, )_{k\in\N}
	\;\;\;,\;\;\;
	\opB^- \Gamma :=( \, B^-_{k+1} \gamma^{(k+1)} \, )_{k\in\N} \,.
\end{align*}
Moreover, we define the free evolution operator $U(t)$ by
\begin{align*}
(U(t)\Gamma)^{(k)}=U^{(k)}(t)\gamma^{(k)},
\end{align*}
where
\begin{align*}
    (U^{(k)}(t)\gamma^{(k)})(\underline{x}_k,\underline{x}_k')
    =e^{it\Delta_{\underline{x}_k}}e^{-it\Delta_{\underline{x}_k'}}\gamma^{(k)}(\underline{x}_k,\underline{x}_k')
\end{align*}
corresponds to the $k$-th component.

\subsection{The BBGKY hierarchy} 
\label{ssec-BBGKY-1}
 
The cubic defocusing BBGKY hierarchy in $\mathbb{R}^3$ is given by
\begin{align}
	i\partial_t \gammaN^{(k)}(t) & =\sum_{j=1}^k [-\Delta_{x_j},\gammaN^{(k)}(t)]
	+ \frac{1}{N }\sum_{1\leq j< k} [\VN(x_j-x_k  ) , \gammaN^{(k)}(t) ] 
	\nonumber\\
	&\hspace{0.5cm}+\frac{(N-k)}{N }
	\sum_{1\leq  j\leq k}\tr_{k+1} [\VN(x_j-x_{k+1} ) , \gammaN^{(k+1)}(t) ] \,,
	\label{bbgky}
\end{align}  
for $k=1,\dots,N$. 

We extend this finite hierarchy trivially to an infinite hierarchy by
adding the terms  $\gammaN^{(k)}=0$ for $k > N$, and write
\begin{align} \label{eq-def-BBGKYN-1}
	i\partial_t \gammaN^{(k)} = \sum_{j=1}^k [-\Delta_{x_j},\gammaN^{(k)}]   
	+  (B_{N} \GammaN)^{(k)}
\end{align}
for $k\in\N$. Here, we have $\gammaN^{(k)}=0$ for $k>N$, and we define
\begin{align} \label{eq-def-BBGKYN-Bop-1}
	(B_N\Gamma_N)^{(k)}:=
	\begin{cases}
	B_{N;k+1}^{main} \gammaN^{(k+1)} + B_{N;k}^{error} \gammaN^{(k)} & {\rm if} \; k\leq N \\
	& \\
	0 & {\rm if} \; k>N.
	\end{cases}
\end{align}
The interaction terms on the right hand side are defined by
\begin{align}
	B_{N;k+1}^{main}\gamma_N^{(k+1)}
	= B^{+,main}_{N;k+1}\gamma_N^{(k+1)}
        - B^{-,main}_{N;k+1}\gamma_N^{(k+1)} \, ,
\end{align}
and 
\begin{align} 
	B_{N;k}^{error}\gamma_N^{(k)}
	= B^{+,error}_{N;k}\gamma_N^{(k)}
        - B^{-,error}_{N;k}\gamma_N^{(k)} \, ,
\end{align}
where 
\begin{align}
	B^{\pm,main}_{N;k+1}\gamma_N^{(k+1)}:= 
	\frac{N-k}{N} \sum_{j=1}^k B^{\pm,main}_{N;j;k+1}\gamma_N^{(k+1)},
\end{align}
and 
\begin{align}
	B^{\pm,error}_{N;k}\gamma_N^{(k)}:= 
	\frac{1}{N} \sum_{i<j}^{k} B^{\pm,error}_{N;i,j;k}\gamma_N^{(k)},
\end{align}                  
with 
\begin{align}	 
	 &\Big( B^{+,main}_{N;j;k+1 }\gamma_N^{(k+1)}\Big)
	(t,x_1,\dots,x_k;x_1',\dots,x_k') 
	\nonumber\\
	&
	= \int dx_{k+1}  
	 \VN(x_j-x_{k+1}) 
        \gamma_N^{(k+1)} (t,x_1,\dots,x_{k },x_{k+1};x_1',\dots,x_{k }',x_{k+1})
        \quad\quad
        \label{eq-Bmain-def-1}
\end{align}
and 
\begin{align}
	&\Big(B^{+,error}_{N;i,j;k}\gamma_N^{(k)}\Big)
	(t,x_1,\dots,x_k;x_1',\dots,x_k') 
	\nonumber\\
	&\hspace{2cm}
	= \VN(x_i-x_{j}) 
        \gamma^{(k)}(t,x_1,\dots,x_{k };x_1',\dots,x_{k }') \,.
        \label{eq-Berror-def-1}
\end{align}
Moreover,
\begin{align} 
& \left(B^{-,main}_{N;j;k+1}\gamma_N^{(k+1)}\right)
(t,x_1,\dots,x_k;x_1',\dots,x_k') \\
& \quad \quad = \int dx_{k+1} 
	 \VN(x'_j-x_{k+1}) 
        \gamma_N^{(k+1)}(t,x_1,\dots,x_{k},x_{k+1};x_1',\dots,x_{k}',x_{k+1}).
        \label{eq-Berror-def-1-2}
\end{align} 
and 
\begin{align*} 
& \left(B^{-,error}_{N;i,j;k}\gamma_N^{(k)}\right)
(t,x_1,\dots,x_k;x_1',\dots,x_k') \\
& \quad \quad =  \VN(x'_i-x'_{j}) 
        \gamma^{(k)}(t,x_1,\dots,x_{k };x_1',\dots,x_{k }')\, .
\end{align*}
We remark that in all of the above definitions, we have that 
$B^{\pm,main}_{N;k}$, $B^{\pm,error}_{N;k}$, etc. are 
defined to be given by multiplication with zero for $k>N$.
 
Then, we can write the BBGKY hierarchy compactly in the form
\begin{align} \label{eq-BBGKY-condensed}
        i\partial_t \GammaN + \opDelta_\pm \GammaN & = \opBN \GammaN 
        \nonumber\\
        \GammaN(0) &\in \cH^\alpha_\xi  \,,
\end{align}
where
$$
	\opDelta_\pm \GammaN := ( \, \Delta^{(k)}_\pm \gammaN^{(k)} \, )_{k\in\N} \, ,
        \; \; \; \; \mbox{ with }
        \Delta_{\pm}^{(k)} = \sum_{j=1}^{k} \left( \Delta_{x_j} - \Delta_{x'_j} \right)\, ,
$$
and 
\begin{align} \label{eq-BBGKY-condensed-B} 
	\opBN \GammaN := ( \, B_{N;k+1} \gammaN^{(k+1)} \, )_{k\in\N} \,.
\end{align}
In addition, we introduce the notation
\begin{align*} 
	& \opBN^+ \GammaN := \, ( \, B^+_{N;k+1} \gammaN^{(k+1)} \, )_{k\in\N}
        \nonumber \\  
	& \opBN^- \GammaN := \, ( \, B^-_{N;k+1} \gammaN^{(k+1)} \, )_{k\in\N} \,
\end{align*}
adapted to \eqref{eq-Berror-def-1} and \eqref{eq-Berror-def-1-2}.

$\;$\\

\section{Derivation of GP from BBGKY hierarchy}
\label{positivity_proof_section}

In this section, we prove Theorem \ref{derivation}.

\subsection{Local well-posedness of the BBGKY hierarchy in $\mathcal{H}^1_\xi$}

Taking $\delta=0$ in Lemma 4.1 of \cite{CPBBGKY} gives us local well-posedness of the $K$-truncated $N$-BBGKY hierarchy:
\begin{lemma}\label{KN_BB_LWP}
Let $K<b_1 \log N$, for some constant $b_1>0$.  
Let 
\eqn 
    \Gamma_{0,N}^K:=P_{\le K}\Gamma_{0,N}=(\gamma_{0,N}^{(k)})_{k=1}^K\in \cH_{\xi'}^1\,,
\eeqn 
and 
\begin{align}\label{T_0_def}
T_0(\xi):=\xi^2/c_0\;,\;\;\xi\in\R_+\,,
\end{align}
where the constant $c_0>0$ is defined in Lemma \ref{lm-boardgame-est-1}.
Then, for $0<\xi'<1$ and $\xi$ satisfying \eqref{xi_relationship}, the following holds:

For $0<T<T_0(\xi)$, and for $b_1>0$ sufficiently small (see \eqref{eq-b1-def-1} below), 
there exists a unique solution 
$\Gamma_N^K\in L^\infty_{t\in I}\cH_\xi^1$ of the BBGKY hierarchy \eqref{bbgky}
for $I=[0,T]$ such that $\opB_N\Gamma_N^K\in L^2_{t\in I}\cH_\xi^1$.
Moreover,
\begin{align}
	\|\Gamma_N^K\|_{L^\infty_{t\in I}\cH_\xi^1}
	\, \leq \,
	 C_0(T,\xi,\xi') \, \|\Gamma_{0,N}^K\|_{\cH_{\xi'}^1}
\end{align}
and
\begin{align}\label{eq-ThetaNK-L2Hxi-bound-0-1}
	\|\opB_N\Gamma_N^K\|_{L^2_{t\in I}\cH_\xi^1} 
 	\, \leq \, C_0(T,\xi,\xi') \, \|\Gamma_{0,N}^K\|_{\cH_{\xi'}^1}
\end{align}
hold. The constant $C_0=C_0(T,\xi,\xi')$ is
independent of $N$.

Furthermore, $(\Gamma_N^K)^{(k)}=0$ for all $K<k\le N$, and $t\in I$.
\end{lemma}

\subsection{Constants}\label{ssec-constants-1}
For convenience, we collect here the interdependence of various constants that appear in the
formulation of the main theorems above. 
Throughout this paper, we will require that, given $\xi'>0$, 
the real, positive constants $\xi$ and $\xi_1$ satisfy
\begin{align}
\label{xi_relationship}
\begin{cases}
\xi<\eta\min\left\{\frac{1}{\xi'}\,e^{-2\beta/b_1}, e^{-24\beta/b_1}\right\}\qquad\text{and}\\
0<\xi_1<\theta^3\xi<\theta^6\xi',
\end{cases}
\end{align}
where $\theta:=\min\{\eta,(1+\tfrac{2}{5}C_{Sob})^{-2/5}\}$;
the constant $\eta>0$ is defined in Lemma \ref{lm-BGamma-Cauchy-1}, and $C_{Sob}>0$ is the constant in the 
trace Sobolev inequality 
\begin{align} 
  \Big(\int dx |f(x,x)|^2\Big)^{\frac12} \, \leq \, C_{Sob}
  \Big(\int dx_1 \, dx_2 \, 
	\left| \, \bra\nabla_{x_1}\ket  \bra\nabla_{x_2}\ket  \,
	f(x_1, x_2)\right|^2\Big)^{\frac12} 
\end{align}
for  $x_{1,2}\in\R^3$, see   \cite{CPHE}.  This will ensure that $\xi_1$ and $\xi$ are small enough so that the results of both  \cite{CPBBGKY} and \cite{CPHE} hold.

In \eqref{xi_relationship},  $b_1>0$ is a constant chosen sufficiently small that Lemma \ref{lm-BGamma-Cauchy-1}  holds for all $K,N$
satisfying
\eqn\label{eq-b1-def-1}
    K \le b_1 \log N\,.
\eeqn
To satisfy this requirement, $b_1$ only depends on $\beta$ (see \eqref{betadef}).

$\;$ \\

\subsection{From $(K,N)$-BBGKY to $K$-truncated GP hierarchy}
In this section, we show that solutions to the $(K,N)$-BBGKY hierarchy approach those of the $K$-Truncated GP hierarchy as $N\rightarrow\infty$.
 
\begin{proposition}\label{prp-GammaKN-conv-1} 
Suppose that $\Gamma_0=(\gamma_0^{(k)})_{k=1}^\infty\in\mathfrak{H}^1_{\xi'}$.  
Moreover, let 
$\Gamma^K\in \{\Gamma\in L^\infty_{t\in [0,T]}\mathcal{H}^1_\xi \, | \, 
B\Gamma\in L^2_{t\in [0,T]}\mathcal{H}^1_\xi\}$ be the solution of the GP hierarchy \eqref{gp} with truncated initial data
$\Gamma_0^K=P_{\leq K}\Gamma_0$ constructed in \cite{chpa4}, where $0<\xi'<1$ and $\xi$ satisfy \eqref{xi_relationship}, and $0<T<T_0(\xi)$  (see \eqref{T_0_def}).  Let $\Gamma_N^K$ solve the $(K,N)$-BBGKY hierarchy \eqref{bbgky} with the same initial data
$\Gamma_{0,N}^K:=P_{\le K}\Gamma_0$.  Let
\begin{align}
K(N)\le b_1\log N\label{K def}
\end{align}
as in Lemma \ref{KN_BB_LWP}.  Then,
\begin{align}\label{eq-GammaK-Nlim-1}
	\lim_{N\rightarrow\infty}\| \, \Gamma_N^{K(N)} 
	- \Gamma^{K(N)} \, \|_{L^\infty_{t\in[0,T]}\cH_\xi^1} = 0
\end{align}
and
\begin{align}\label{eq-BGammaK-Nlim-1}
	 \lim_{N\rightarrow\infty} \| \, \opB_{N}\Gamma_N^{K(N)} 
	 \,  - \, \opB\Gamma^{K(N)} \, \|_{L^2_{t\in[0,T]}\cH_\xi^1}
	 = 0.
\end{align}
\end{proposition}

\prf
In \cite{chpa4}, the authors constructed a solution $\Gamma^K$ of the full GP hierarchy
with truncated initial data, $\Gamma(0)=\Gamma_0^K\in\cH_{\xi'}^{1}$, such that
for an arbitrary fixed $K$,  $\Gamma^K$ satisfies the GP-hierarchy  in  integral 
representation, 
\begin{align} \label{eq-outl-GP-Duhamel-special-1}
	\Gamma^K(t) = U(t)\Gamma^K_0 + i \,  \int_0^t U(t-s) \, \opB\Gamma^K(s) \, ds \,.
\end{align}
and, in particular, $(\Gamma^K)^{(k)}(t)=0$ for all $k>K$.

Accordingly, we have
\begin{align}
	&\opB_N\Gamma_N^K-\opB\Gamma^K
	\nonumber\\
	&=
	\opB_N U(t)\Gamma_{0,N}^K-\opB U(t) \Gamma_0^K 
	\nonumber\\
	&\hspace{0.5cm}
	+\, i   \int_0^t \, \big( \, \opB_N U(t-s) \opB_N\Gamma_N^K
	- \opB U(t-s) \opB \Gamma^K  \, \big)(s) ds \, 
	\nonumber\\
	&=
	(\opB_N - \opB) U(t) \Gamma_{0,N}^K 
	+  \opB U(t)(\Gamma_{0,N}^K- \Gamma_0^K) 
	\nonumber\\
	&\hspace{0.5cm}
	+\,  i  \, \int_0^t \big( \, \opB_N -\opB \big) U(t-s) \opB\Gamma^K (s) \, ds 
	\nonumber\\
	&\hspace{0.5cm}
	+\, i  \int_0^t \opB_N  U(t-s) \big( \opB_N\Gamma_N^K -  
	\opB\Gamma^K  \, \big)(s) \, ds \,.
\end{align}
Here, we observe that we can apply Lemma \ref{lm-BGamma-Cauchy-1} with
\begin{align}
	\widetilde\Theta_N^K := \opB_N\Gamma_N^K-\opB\Gamma^K
\end{align}
and
\begin{align}
	\Xi_N^K &:= (\opB_N - \opB) U(t) \Gamma_{0,N}^K 
	+  \opB U(t)(\Gamma_{0,N}^K- \Gamma_0^K) 
	\nonumber\\
	&\hspace{0.5cm}
	+i\int_0^t \big( \, \opB_N -\opB \big) U(t-s) \opB\Gamma^K (s) \, ds \,.
\end{align}
Given $\xi'$, we introduce  parameters $\xi,\xi'',\xi'''$ satisfying
\begin{align}\label{eq-xiparm-def-1}
	\xi \, < \, \theta \, \xi'' \, <  \, \theta^2 \, \xi''' \, <  \, \theta^3 \xi' 
\end{align}
where the constant $\theta$ is defined as in \eqref{xi_relationship}, so that $0<\theta\le \eta$, where $\eta$ is defined as in  Lemma \ref{lm-BGamma-Cauchy-1}.
Accordingly,  Lemma \ref{lm-BGamma-Cauchy-1} implies that
\begin{align}
	&\|\opB_N\Gamma_N^K-\opB\Gamma^K\|_{L^2_{t\in [0,T]}\cH^1_\xi}
	\nonumber\\ 
	&\leq
	C_0(T,\xi,\xi'') 
	\Big( \,  \|B U(t)(\Gamma_{0,N}^K- \Gamma_0^K)\|_{L^2_{t\in [0,T]}\cH^1_{\xi''}}
	+ R_1(N)+R_2(N) \, \Big)
	\nonumber\\ 
	&\leq
	C_1(T,\xi,\xi',\xi'') 
	\Big( \, \| \Gamma_{0,N}^K- \Gamma_0^K\|_{L^2_{t\in [0,T]}\cH^1_{\xi'}}
	+ R_1(N)+R_2(N) \, \Big)  \,,\label{terms}
\end{align}
where we used Lemma A.1 in \cite{CPBBGKY} to pass to the last line. Here,
\begin{align}
	R_1(N) :=
	\|(\opB_N - \opB) U(t) \Gamma_{0,N}^K\|_{L^2_{t\in [0,T]}\cH^1_{\xi''}}
\end{align}
and
\begin{align}
	R_2(N) :=  
	\Big\| \, \int_0^t \big( \, \opB_N -\opB \big) U(t-s) \opB\Gamma^K (s) \, ds
	\, \Big\|_{L^2_{t\in [0,T]}\cH^1_{\xi''}} \,.
\end{align}
Next, we consider the limit $N\rightarrow\infty$ with $K(N)$ as given in \eqref{K def}.

We have 
\begin{align}
	\lim_{N\rightarrow \infty}
	\| \Gamma_{0,N}^{K(N)}- \Gamma_0^{K(N)}\|_{\cH^{1}_{\xi'}}
	&=
	\lim_{N\rightarrow \infty}
	\| \, P_{\leq K(N)} \, ( \, \Gamma_{0,N}- \Gamma_0 \, ) \, \|_{\cH^{1}_{\xi'}}
	\nonumber\\
	&\leq
	\lim_{N\rightarrow \infty}
	\| \,   \Gamma_{0,N}- \Gamma_0 \, \|_{\cH^{1}_{\xi'}}
	\nonumber\\
	&=
	0.
\end{align}
By Lemmas \ref{lm-R1bound-1-new} and \ref{lm-R1bound-2-new} below, we have that

\begin{align*}
	\lim_{N\rightarrow\infty} R_1(N)&=0
\end{align*}
and
\begin{align*}
  \lim_{N\rightarrow\infty} R_2(N)&=0.
\end{align*}
Thus $\eqref{terms}\rightarrow 0$ as $N\rightarrow\infty$, and hence the limit \eqref{eq-BGammaK-Nlim-1} holds.  To prove \eqref{eq-GammaK-Nlim-1}, we observe that
\begin{align*}
&\Gamma_N^{K(N)}(t)-\Gamma^{K(N)}(t)\\
&=U(t)\left(\Gamma_N^{K(N)}(0)-\Gamma^{K(N)}(0)\right)+i\int_0^tU(t-s)\left(B_N\Gamma_N^{K(N)}(s)-B\Gamma^{K(N)}(s)\right)ds,
\end{align*}
and hence, for $0<t<T$,
\begin{align}
&\|\Gamma_N^{K(N)}(t)-\Gamma^{K(N)}(t)\|_{\mathcal{H}^1_\xi}\nonumber\\
&\le\|U(t)\left(\Gamma_N^{K(N)}(0)+\Gamma^{K(N)}(0)\right)\|_{\mathcal{H}^1_\xi}
\nonumber\\
&\hspace{3cm}
+t^{1/2}\|U(t-s)\left(B_N\Gamma_N^{K(N)}(s)-B\Gamma^{K(N)}(s)\right)\|_{L^2_{s\in [0,t]}\mathcal{H}^1_\xi}\nonumber\\
&=\|\Gamma_N^{K(N)}(0)+\Gamma^{K(N)}(0)\|_{\mathcal{H}^1_\xi}+t^{1/2}\|B_N\Gamma_N^{K(N)}-B\Gamma^{K(N)}\|_{L^2_{[0,t]}\mathcal{H}^1_\xi}\nonumber\\
&\le\|\Gamma_N^{K(N)}(0)+\Gamma^{K(N)}(0)\|_{\mathcal{H}^1_\xi}+T^{1/2}\|B_N\Gamma_N^{K(N)}-B\Gamma^{K(N)}\|_{L^2_{[0,T]}\mathcal{H}^1_\xi}\label{last}\\
&\rightarrow 0\text{ as }N\rightarrow\infty\text{ by }\eqref{eq-BGammaK-Nlim-1}.\nonumber
\end{align}
Since the last line \eqref{last} is independent of $t$, the result \eqref{eq-GammaK-Nlim-1} follows.
\endprf

\begin{lemma}\label{lm-R1bound-1-new}
Under the same assumptions as in Proposition \ref{prp-GammaKN-conv-1},
\begin{align}
\lim_{N\rightarrow\infty}\left\|(B_N-B)U(t)\,\Gamma_{0,N}^{K(N)}\right\|_{L^2_{t\in\mathbb{R}}\mathcal{H}^1_{\xi''}}=0.\label{quantity}
\end{align}
\end{lemma}

\begin{proof} 
We recall that for $g:\mathbb{R}^n\rightarrow\mathbb{C}$ of the form $g(x)=f(x,x)$ for some Schwartz class fuction 
$f:\mathbb{R}^n\times\mathbb{R}^n\rightarrow\mathbb{C}$, one has
\begin{align}
\widehat g(\xi)=\int\widehat f(\xi-\eta,\eta)\,d\eta.\label{ft}
\end{align} 
We note that the Fourier transform of $\left(U^{(k+1)}(t)\gamma^{(k+1)}_0\right)(\underline{x}_{k+1},\underline{x}_{k+1}')$ with respect to the variables $(t,\underline{x}_{k+1},\underline{x}_{k+1}')$ is given by
\begin{align}
\delta(\tau+|\underline{\xi}_{k+1}|^2-|\underline{\xi}_{k+1}'|^2)\widehat{\gamma}_0^{(k+1)}(\underline{\xi}_{k+1},\underline{\xi}_{k+1}').\label{free_Fourier}
\end{align}

Recall that $B_{N,1,k+1}^{+,main}U(t)(\gamma_{0,N}^{(k+1)})$ is given by
\begin{align*}
&\int V_N(x_1-x_{k+1})\gamma_N^{(k+1)}(t,x_1,...,x_k,x_{k+1};x_1',...,x_k',x_{k+1})dx_{k+1}
\end{align*}
and hence its Fourier transform with respect to the variables $(t,\underbar{x}_k,\underbar{x}_k')$ is given by
\begin{align}
&\int  e^{-ix_{k+1}\fvar_1}\widehat{V}_N(\fvar_1)*_{\fvar_1}(F\gamma_N^{(k+1)})(\tau,\fvar_1,...,\fvar_k,x_{k+1};\fvar_1',...,\fvar_k',x_{k+1})\,dx_{k+1}\nonumber\\
&=\int\int e^{-ix_{k+1}\eta}\widehat{V}_N(\eta)(F\gamma_N^{(k+1)})(\tau,\fvar_1-\eta,\fvar_2,...,\fvar_k,x_{k+1};\fvar_1',...,\fvar_k',x_{k+1})\,d\eta \,dx_{k+1}\nonumber\\
&=\int\int\widehat{V}_N(\eta)\widehat{\gamma}_N^{(k+1)}(\tau,\fvar_1-\eta,\fvar_2,...,\fvar_k,\eta-\nu;\fvar_1',...,\fvar_k',\nu)\,d\nu\,d\eta\qquad\text{(by \eqref{ft})}\nonumber\\
&=\int\int\widehat{V}_N(\eta+\nu)\widehat{\gamma}_N^{(k+1)}(\tau,\fvar_1-\eta-\nu,\fvar_2,...,\fvar_k,\eta;\fvar_1',...,\fvar_k',\nu)\,d\nu\,d\eta\nonumber
\end{align}
where we substituted $\eta\rightarrow\eta+\nu$. Thus, the above equals
\begin{align}
&=\int\int\widehat{V}_N(\fvar_{k+1}+\fvar_{k+1}') \nonumber\\
&\hspace{1cm}
\widehat{\gamma}_N^{(k+1)}(\tau,\fvar_1-\fvar_{k+1}-\fvar_{k+1}',\fvar_2,...,\fvar_k,\fvar_{k+1};\fvar_1',...,\fvar_k',\fvar_{k+1}')\,d\fvar_{k+1}'\,d\fvar_{k+1}\nonumber\\
&=\int\int\widehat{V}_N(\fvar_{k+1}+\fvar_{k+1}')\delta(\cdots)\nonumber\\
&\hspace{1cm}
\widehat{\gamma}^{(k+1)}_{0,N}(\fvar_1-\fvar_{k+1}-\fvar_{k+1}',\fvar_2,...,\fvar_k,\fvar_{k+1};\underline{\fvar}_{k+1}')\,d\fvar_{k+1}'\,d\fvar_{k+1}\label{lastt}
\end{align}
where the operator $F$ is the Fourier transform with respect to the variables $(t,\underbar{x}_k,\underbar{x}_k')$ and  
\begin{align*}
\delta(...):=\delta(\tau+|\fvar_1-\fvar_{k+1}-\fvar_{k+1}'|^2+|\underline{\fvar}_{k+1}|^2-|\fvar_1|^2-|\underline{\fvar}_{k+1}'|^2).
\end{align*}
Equation \eqref{free_Fourier} was used to pass to the last line \eqref{lastt}.  Similarly, the Fourier transform of $B_{k+1}^+U(t)(\gamma_{0,N}^{(k+1)})$ with respect to the variables $(t,\underbar{x}_k,\underbar{x}_k')$ is given by
\begin{align*}
\int\int\delta(...)\widehat\gamma_{0,N}^{(k+1)}(\fvar_1-\fvar_{k+1}-\fvar_{k+1}',\fvar_2,...,\fvar_{k+1},\underline{\fvar}_{k+1}')d\fvar_{k+1}d\fvar_{k+1}'.
\end{align*}

Thus,
\begin{align}
&\|(B_{N;1;k+1}^{+,main}-B_{1;k+1}^+)U(t)\gamma_{0,N}^{(k+1)}\|_{L^2_t{H}^1}^2\nonumber\\
&=\int\int\int\prod_{j=1}^k\langle\fvar_j\rangle^2\prod_{j=1}^k\langle\fvar_j'\rangle^2\bigg(\int\int(1-\widehat{V}_N(\fvar_{k+1}+\fvar_{k+1}'))\delta(...)\nonumber\\
&\hspace{1cm}\widehat{\gamma}^{(k+1)}_{0,N}(\fvar_1-\fvar_{k+1}-\fvar_{k+1}',\fvar_2,...,\fvar_k,\fvar_{k+1};\underline{\fvar}_{k+1}')\,d\fvar_{k+1}'\,d\fvar_{k+1}\bigg)^2\,d\underline{\fvar}_k\,d\underline{\fvar}_k'\,d\tau\nonumber\\
&\le\int\int\int J(\tau,\underline{\fvar}_k,\underline{\fvar}_k')\int\int\delta(...)\langle \fvar_1-\fvar_{k+1}-\fvar_{k+1}'\rangle^2\langle \fvar_{k+1}\rangle^2\langle \fvar_{k+1}'\rangle^2
\nonumber\\
&\hspace{1cm}\prod_{j=2}^k\langle\fvar_j\rangle^2\prod_{j'=1}^k\langle\fvar_j'\rangle^2|1-\widehat{V}_N(\fvar_{k+1}+\fvar_{k+1}')|^2
\nonumber\\
&\hspace{3cm}
|\widehat{\gamma}^{(k+1)}_{0,N}(\fvar_1-\fvar_{k+1}-\fvar_{k+1}',\fvar_2,...,\fvar_k,\fvar_{k+1};\underline{\fvar}_{k+1}')|^2\nonumber\\
&\hspace{7cm}\,d\fvar_{k+1}'\,d\fvar_{k+1}\,d\underline{\fvar}_k\,d\underline{\fvar}_k'\,d\tau\label{integral}
\end{align}
where
\begin{align*}
J(\tau,\underline{\fvar}_k,\underline{\fvar}_k')
:=\int\int\frac{\delta(...)\langle \fvar_1\rangle^2}{\langle\fvar_1-\fvar_{k+1}-\fvar_{k+1}'\rangle^2\langle\fvar_{k+1}\rangle^2\langle\fvar_{k+1}'\rangle^2}\,d\fvar_{k+1} d\fvar_{k+1}'
\end{align*}
and $J(\tau,\underline{\fvar}_k,\underline{\fvar}_k')$ is  bounded uniformly in $\tau,\underline{\fvar}_k,\underline{\fvar}_k'$,
see Proposition 2.1 of \cite{KM}.  
 
Let $\delta$ satisfy $0<\delta<\beta$.  Recall that $\widehat{V}_N(u)=\widehat V(N^{-\beta}u)$.  The integral \eqref{integral} can now be separated into the regions $\{|\fvar_{k+1}+\fvar_{k+1}'|<N^\delta\}$ and $\{|\fvar_{k+1}+\fvar_{k+1}'|\ge N^\delta\}$.

The portion of the integral \eqref{integral} over $\{|\fvar_{k+1}+\fvar_{k+1}'|<N^\delta\}$ is bounded by
\begin{align}
C_VN^{4(\delta-\beta)}\|\gamma_{0,N}^{(k+1)}\|_{H^1}^2\label{bound1}
\end{align}
because $\nabla \widehat V(0)=0$ and $\widehat V\in C^2$, so by bounding the Taylor remainder term,
\begin{align} &\sup_{|\fvar_{k+1}+\fvar_{k+1}'|<N^\delta}|1-\widehat{V}_N(\fvar_{k+1}+\fvar_{k+1}')|^2\nonumber\\
&=\sup_{|\fvar_{k+1}+\fvar_{k+1}'|<N^\delta}|1-\widehat{V}(N^{-\beta}(\fvar_{k+1}+\fvar_{k+1}'))|^2\nonumber\\
&\le\sup_{|\fvar_{k+1}+\fvar_{k+1}'|<N^\delta}C_V(N^{-\beta}(\fvar_{k+1}+\fvar_{k+1}'))^4\nonumber\\
&\le C_VN^{4(\delta-\beta)},\nonumber
\end{align}
where $C_V$ is the $L^\infty$ norm of the second derivative of $V$.

The portion of the integral \eqref{integral} over $\{|\fvar_{k+1}+\fvar_{k+1}'|\ge N^\delta\}$ is bounded by
\begin{align}
&
a_{k,N}^2
\nonumber\\
&:=\int\int\int J(\tau,\underline{\fvar}_k,\underline{\fvar}_k')\int_{|\fvar_{k+1}|\ge N^\delta}\int\delta(...)\langle \fvar_1-\fvar_{k+1}-\fvar_{k+1}'\rangle^2\nonumber\\
&\hspace{1cm}
\langle \fvar_{k+1}\rangle^2\langle \fvar_{k+1}'\rangle^2
\prod_{j=2}^k\langle\fvar_j\rangle^2\prod_{j'=1}^k\langle\fvar_j'\rangle^2(1+\|\widehat{V}\|_\infty)^2
\nonumber\\
&\hspace{3cm}
\Big|\widehat{\gamma}^{(k+1)}_{0,N}(\fvar_1-\fvar_{k+1}-\fvar_{k+1}',\fvar_2,...,\fvar_k,\fvar_{k+1};\underline{\fvar}_{k+1}')
\Big|^2\nonumber\\
&\hspace{7cm}\,d\fvar_{k+1}'\,d\fvar_{k+1}\,d\underline{\fvar}_k\,d\underline{\fvar}_k'\,d\tau\label{akn}\\
&\le C\|\gamma_{0,N}^{(k+1)}\|_{H^1}^2\nonumber\\
&= C\|\gamma_{0}^{(k+1)}\|_{H^1}^2.\label{bound2}
\end{align} 

We are now ready to bound the desired qauntity \eqref{quantity} in the statement of the lemma.  

Let $\Omega_{k,N}=\{|\fvar_{k+1}+\fvar_{k+1}'|<N^\delta\}$.  Then,
\begin{align}
&\|(B_N^+-B^+)U(t)\Gamma_{0,N}^K\|_{L^2_{t\in \mathbb{R}}\mathcal{H}^1_{\xi''}}-\underbrace{\|B_N^{+,error}U(t)\Gamma_{0,N}^K\|_{L^2_{t\in\mathbb{R}}\mathcal{H}^1_{\xi''}}}_{\rightarrow 0\text{ as }N\rightarrow\infty\text{ by Proposition A.2 in \cite{CPBBGKY}}}\nonumber\\
&\le\sum_{k=1}^K\sum_{j=1}^k(\xi'')^k\|(B_{N;j;k+1}^{+,main}-B_{j,k+1}^+)U^{(k+1)}(t)\gamma_{0,N}^{(k+1)}\|_{L^2_{t\in\mathbb{R}}H^1}\nonumber\\
&\le\sum_{k=1}^Kk(\xi'')^k\|(B_{N;1;k+1}^{+,main}-B_{1;k+1}^+)U^{(k+1)}(t)\gamma_{0,N}^{(k+1)}\|_{L^2_{t\in\mathbb{R}}H^1}\nonumber\\
&\le\sum_{k=1}^Kk(\xi')^k(\xi''/\xi')^k\|(B_{N;1;k+1}^{+,main}-B_{1;k+1}^+)U^{(k+1)}(t)\gamma_{0,N}^{(k+1)}\|_{L^2_{t\in\mathbb{R}}H^1}\nonumber\\
&\le\left(\sup_k k(\xi''/\xi')^k\right)\bigg(\sum_{k=1}^K(\xi')^k\|(B_{N;1;k+1}^{+,main}-B_{1;k+1}^+)U^{(k+1)}(t)\gamma_{0,N}^{(k+1)}\|_{L^2_{t\in\mathbb{R}}H^1(\Omega_{k,N})}\nonumber\\
&\hspace{.5cm}+\sum_{k=1}^K(\xi')^k\|(B_{N;1;k+1}^{+,main}-B_{1;k+1}^+)U^{(k+1)}(t)\gamma_{0,N}^{(k+1)}\|_{L^2_{t\in\mathbb{R}}H^1(\Omega_{k,N}^c)}\bigg)\nonumber\\
&\le\left(\sup_k k(\xi''/\xi')^k\right)\bigg( C_V(1+\|\widehat{V}\|_\infty)\sum_{k=1}^K(\xi')^kN^{2(\delta-\beta)}\|\gamma_{0,N}^{k+1}\|_{H^1}+\sum_{k=1}^K (\xi')^k a_{k,N}\bigg),\label{sum}
\end{align}
where \eqref{bound1} and \eqref{bound2} were used to pass to the last line \eqref{sum}.

Now, for $(k,N)\in\mathbb{N}\times\mathbb{N}$, we define
\begin{align}
&\widetilde {a_{k,N}}^2\nonumber\\
&:=\int\int\int J(\tau,\underline{\fvar}_k,\underline{\fvar}_k')\int_{|\fvar_{k+1}+\fvar_{k+1}'|\ge N^\delta}\int\delta(...)\langle \fvar_1-\fvar_{k+1}-\fvar_{k+1}'\rangle^2\langle \fvar_{k+1}\rangle^2\langle \fvar_{k+1}'\rangle^2\nonumber\\
&\hspace{1cm}\prod_{j=2}^k\langle\fvar_j\rangle^2\prod_{j'=1}^k\langle\fvar_j'\rangle^2(1+\|\widehat{V}\|_\infty)^2\nonumber\\
&\hspace{2cm}
\Big|\widehat{\gamma}^{(k+1)}_{0}(\fvar_1-\fvar_{k+1}-\fvar_{k+1}',\fvar_2,...,\fvar_k,\fvar_{k+1};\underline{\fvar}_{k+1}')\Big|^2\nonumber\\
&\hspace{7cm}\,d\fvar_{k+1}'\,d\fvar_{k+1}\,d\underline{\fvar}_k\,d\underline{\fvar}_k'\,d\tau,\label{akn2}
\end{align}
and observe that   $\widetilde {a_{k,N}}=a_{k,N}$ (as defined in \eqref{akn}), for $k\le N$, because $\gamma_{0,N}^{(k)}=\gamma_{0}^{(k)}$
for $k\le N$.  Thus we have that
\begin{align}
\eqref{sum}
&\le\left(\sup_k k(\xi''/\xi')^k\right)
\nonumber\\
&\hspace{2cm}\bigg( C_V(1+\|\widehat{V}\|_\infty)\sum_{k=1}^\infty(\xi')^kN^{2(\delta-\beta)}\|\gamma_{0}^{(k+1)}\|_{H^1}+\sum_{k=1}^\infty (\xi')^k a_{k,N}\bigg)\nonumber\\
&\le\left(\sup_k k(\xi''/\xi')^k\right)\bigg( C_V(1+\|\widehat{V}\|_\infty)N^{2(\delta-\beta)}\|\Gamma_{0}^{(k+1)}\|_{\mathcal{H}^1_{\xi'}}+\sum_{k=1}^\infty (\xi')^k a_{k,N}\bigg).\label{sum2}
\end{align}

It follows from the definition \eqref{akn2} of 
$a_{k,N}^2$ that $\sum_{k=1}^\infty (\xi')^k a_{k,N}\le C\|\Gamma_0\|_{\mathcal{H}^1_{\xi'}}$ 
and that, for fixed $k$, $a_{k,N}\searrow0$ monotonically as $N\rightarrow\infty$.  
This is because $a_{k,N}^2$ is an integral where the integrand is independent of $N$ and the region of integration shrinks as $N$ grows.  Thus, by the monotone convergence theorem, 
$\sum_{k=1}^\infty (\xi')^k a_{k,N}\searrow 0$ as $N\rightarrow\infty$.  
Therefore $\eqref{sum2}\rightarrow 0$ as $N\rightarrow\infty$. 
\end{proof}

\begin{lemma}\label{lm-R1bound-2-new}
Under the same assumptions as in Proposition \ref{prp-GammaKN-conv-1},
\begin{align*}
\lim_{N\rightarrow\infty}\left\|\int_0^t(B_N-B)U(t-s)B\Gamma^{K}(s)\,ds\right\|_{L^2_{t\in I}\mathcal{H}^1_{\xi''}} =0.
\end{align*}
\end{lemma}
\begin{proof}
We have that
\begin{align*}
&\left\|\int_0^t(B_N-B)U(t-s)B\Gamma^{K}(s)\,ds\right\|_{L^2_{t\in I}\mathcal{H}^1_{\xi''}}\\
&\le\int_0^T\left\|(B_N-B)U(t-s)B\Gamma^{K}(s)\right\|_{L^2_{t\in I}\mathcal{H}^1_{\xi''}}\,ds.
\end{align*}
By the same arguments as in the proof of Lemma \ref{lm-R1bound-1-new} above, the integral above goes to zero as $N\rightarrow\infty$ provided that $\|U(t-s)B\Gamma^K(s)\|_{L^2_{t\in I}\mathcal{H}^1_{\xi''
}}$ is uniformly bounded in $N$.  See \cite{chpa4} for a proof of the boundedness of $\|U(t-s)B\Gamma^K(s)\|_{L^2_{t\in I}\mathcal{H}^1_{\xi''}}$.
\end{proof}

\subsection{Control of $\Gamma^{\Phi_N}$ and $\Gamma_N^K$ as $N\rightarrow\infty$}\label{control}

We begin by stating an energy estimate used by Erd\"os, Schlein, and Yau in \cite{esy2}.  
We define the notation $R^{(k,\alpha)}:=\prod_{j=1}^k(1-\Delta_{x_j})^{\alpha/2}$.
\begin{prop}\label{energyprop}
Suppose that $\psi$ is symmetric with respect to permutations of its $N$ variables.  Fix $k\in\mathbb{N}$ and $0<C<1$.  Then there is $N_0=N_0(k,C)$ such that
\begin{align*}
\la\psi,(H_N+N)^k\psi\ra\ge C^kN^k\la\psi,R^{(k,2)}\psi\ra
\end{align*}
for all $N>N_0$.
\end{prop}

\begin{prop}\label{7.1}
Suppose that $b_1>0$, $b_1 \log(N)\ge K(N)\ge \tfrac{1}{2}b_1 \log(N)$, and that $\xi>0$ satisfies
\begin{align}
\xi&<\eta\min\left\{\frac{1}{C}\,e^{-8\beta/b_1}, e^{-24\beta/b_1}\right\},\label{xi min}
\end{align}
where
\begin{align}
{\rm Tr}\,S^{(k,1)}\gamma_{N}^{(K)}(0)<C^K.\label{trace}
\end{align}
Then
\begin{align*}
\lim_{N\rightarrow\infty}\|B_N\Gamma_N^{K(N)}-P_{\le K(N)-1}B_N\Gamma^{\Phi_N}\|_{L^2_{t\in I}\mathcal{H}^1_\xi}=0.
\end{align*}
\end{prop}
\begin{proof}
From Lemma 6.1 in \cite{CPBBGKY}, we have that
\begin{align}
\|B_N\Gamma_N^K-P_{\le K-1}B_N\Gamma^{\Phi_N}\|_{L^2_{t\in I}\mathcal{H}^1_\xi}\le C(T,\xi)(\eta^{-1}\xi)^K K\|(B_N\Gamma^{\Phi_N})^{(K)}\|_{L^2_{t \in I}H^1}\label{bound1b}
\end{align}
holds for a finite constant $C(T,\xi)$ independent of $K$, $N$.

It follows immediately from the definition of $V_N$ that
\begin{align*}
\|\widehat{\nabla V_N}\|_{L^1}\le C N^{4\beta}.
\end{align*}
Thus we have that

\begin{align}
&\|(B_N^+\Gamma^{\Phi_N})^{(K)}\|_{L^2_{t\in I}H^1}^2\nonumber\\
&=\int_I dt \int d\underbar{x}_K \,d\underbar{x}_K'\bigg |\sum_{\ell=1}^K\int\left[\prod_{j=1}^k\langle\nabla_{x_j}\rangle\langle\nabla_{x_j'}\rangle\right]V_N(x_\ell-x_{K+1})\nonumber\\
&\hspace{2cm}\Phi_N(t,\underbar{x}_N)\overline{\Phi_N(t,\underbar{x}_K',x_{K+1},\ldots,x_N)}\,dx_{K+1}\ldots dx_{N}\bigg|^2\nonumber\\ 
&\le CT(\|V_N\|_{L^{\infty}}^2+\|\widehat{\nabla V_N}\|_{L^1}^2)K^2\sup_{t\in I}\left(\|R^{(k,1)}\Phi_N\|_{L^2}\|R^{(k,1)}\Phi_N\|_{L^2}\right)^2\nonumber\\
&=CTN^{8\beta}K^2\sup_{t \in I}\left(\text{Tr}(S^{(K,1)}\gamma_N^{(K)}(t))\right)^2.\label{bound2b}
\end{align}

Since $\la\Phi_N(0),H_N^K,\Phi_N(0)\ra<C^kN^K$, it follows from Proposition \ref{energyprop}, that
\begin{align}
{\rm Tr}(\, S^{(K,1)}\gamma_N^{(K)}(t) \,)
&=\la\Phi_N(t),R^{(K,2)}\Phi_N(t)\ra\nonumber\\
&\le\frac{1}{N^kC^k}\la\Phi_N(t),(H_N+N)^k\Phi_N(t)\ra\nonumber\\
&=\frac{1}{N^kC^k}\la\Phi_N(0),(H_N+N)^k\Phi_N(0)\ra\nonumber\\
&\le\frac{1}{N^kC^k}(2^k\la\Phi_N(0),H_N^k\Phi_N(0)\ra+2^kN^k\la\Phi_N(0),\Phi_N(0)\ra)\nonumber\\
&\le C^k.\label{bound3b}
\end{align}

Combining \eqref{bound1b}, \eqref{bound2b}, and \eqref{bound3b} yields
 
\begin{align*}
&\|B_N\Gamma_N^K-P_{\le K-1}B_N\Gamma^{\Phi_N}\|_{L^2_{t\in I}H^1_\xi}\\
&\le C(T,\xi)(\eta^{-1}\xi)^KK\|(B_N\Gamma^{\Phi_N})^{(K)}\|_{L^2_{t\in I}H^1}&\text{by \eqref{bound1b}}\\
&\le C(T,\xi)(\eta^{-1}\xi)^KKCT^{1/2}N^{4\beta}K\sup_{t\in I}\text{Tr}(S^{(K,1)}\gamma_N^{(K)}(t))&\text{by \eqref{bound2b}}\\
&\le C(T,\xi)(\eta^{-1}\xi)^KKCT^{1/2}N^{4\beta}C^K&\text{by \eqref{bound3b}}\\
&\le\widetilde{C}(T,\xi)(\eta^{-1}\xi)^KKN^{4\beta}C^{K}\\
&\rightarrow 0\text{ as }N\rightarrow\infty
\end{align*}
because $K(N)\ge \tfrac{1}{2}b_1 \log(N)$ and $\xi$ satisfies \eqref{xi min}.
\end{proof}

\subsection{Proof of Theorem \ref{derivation}}
We are now ready to conclude the proof of Theorem \ref{derivation}.   
To this end, we recall again the solution $\Gamma^K$ of the GP hierarchy with 
truncated initial data, $\Gamma^K(t=0)=P_{\leq K}\Gamma_0\in\cH_\xi^1$.
In \cite{chpa4}, the authors proved the existence of a solution  $\Gamma^K$
that satisfies the $K$-truncated GP-hierarchy in  integral form, 
\begin{align} \label{eq-outl-GP-Duhamel-special-3}
	\Gamma^K(t) = U(t)\Gamma^K(0) + i \,  \int_0^t U(t-s) \, \opB\Gamma^K(s) \, ds
\end{align}
where $(\Gamma^K)^{(k)}(t) = 0$ for all $k > K$. Moreover, 
it is shown in \cite{chpa4} that this solution satisfies
$\opB\Gamma^K\in L^2_{t\in I}\cH^{1}_{\xi}$, where $I:=[0,T].$

Additionally, the following convergence was proved in \cite{chpa4}:
\begin{enumerate} 
\item[(a)] The limit 
\begin{align}\label{eq-Gamma-GammaK-1}
	\Gamma := \lim_{K\rightarrow \infty} \Gamma^K 
\end{align}
exists in $L^\infty_t \cH^{1}_{\xi}$. \\

\item[(b)] The limit  
\begin{align}\label{eq-Thetalim-GammaK-1}
	\Theta := \lim_{K \rightarrow \infty} \opB \Gamma^K
\end{align}
exists in $L^2_t \cH^{1}_{\xi}$, and in particular,
\begin{align}\label{eq-Thetalim-GammaK-2}
	\Theta = \opB\Gamma \,.
\end{align}

\item[(c)] The limit $\Gamma$ in equation \eqref{eq-Gamma-GammaK-1} satisfies the full GP hierarchy with initial data $\Gamma_0$.
\end{enumerate}

Clearly, we have that
\begin{align}
&\|B\Gamma-B_NP_{\le K(N)}\Gamma^{\Phi_N}\|_{L^2_{t\in I}\mathcal{H}^1_\xi}\nonumber\\
&\le \|B\Gamma-B\Gamma^{K(N)}\|_{L^2_{t\in I}\mathcal{H}^1_\xi}\label{b1}\\
&\hspace{0.5cm}+\|B\Gamma^{K(N)}-B_N\Gamma_N^{K(N)}\|_{L^2_{t\in I}\mathcal{H}^1_\xi}\label{b2}\\
&\hspace{0.5cm}+\|B\Gamma_N^{K(N)}-B_NP_{\le K(N)}\Gamma^{\Phi_N}\|_{L^2_{t\in I}\mathcal{H}^1_\xi}\label{b3}.
\end{align}
 
In the limit $N\rightarrow\infty$, we have that $\eqref{b1}\rightarrow0$ from 
\eqref{eq-Thetalim-GammaK-1} and \eqref{eq-Thetalim-GammaK-2}.  By Proposition \ref{prp-GammaKN-conv-1}, $\eqref{b2}\rightarrow0$.  $\eqref{b3}\rightarrow 0$ follows from Proposition \ref{7.1}.  This is because $\Gamma_0\in\mathfrak{H}^1_{\xi'}$ and hence \eqref{trace} holds.  Therefore,
\begin{align*}
	\lim_{N\rightarrow\infty}\|\opB\Gamma  
	-\opB_N  \Gamma^{\Phi_N}\|_{L^2_{t\in I}\cH_\xi^1} = 0. \,
\end{align*}

Moreover, we have that

\begin{align}
&\|P_{\le K(N)}\Gamma^{\Phi_N}-\Gamma\|_{L^\infty_{t\in I}\mathcal{H}^1_\xi}\nonumber\\
&\le \|P_{\le K(N)}\Gamma^{\Phi_N}-\Gamma_N^{K(N)}\|_{L^\infty_{t\in I}\mathcal{H}^1_\xi}\label{g1}\\
&\hspace{0.5cm}+\|\Gamma^{K(N)}-\Gamma\|_{L^\infty_{t\in I}\mathcal{H}^1_\xi}\label{g2}\\
&\hspace{0.5cm}+\|\Gamma_N^{K(N)}-\Gamma^{K(N)}|_{L^\infty_{t\in I}\mathcal{H}^1_\xi}\label{g3}
\end{align}

By the Duhamel formula, and applying the Cauchy-Schwarz inequality in time, we have
\begin{align*}
\eqref{g1}
&=\|\int_0^t U(t-s)B_N(P_{\le K(N)}\Gamma^{\Phi_N}-\Gamma_N^{K(N)})(s)\,ds\|_{L^\infty_{t\in I}\mathcal{H}^1_\xi}\\
&\le T^{1/2}\|B_N\Gamma_N^{K(N)}-B_NP_{\le K(N)}\Gamma^{\Phi_N}\|_{L^2_{t\in I}\mathcal{H}^1_\xi}\\
&\rightarrow 0\text{ as }N\rightarrow\infty\text{ by Proposition \ref{7.1}.}
\end{align*}

$\eqref{g2}\rightarrow 0$ as $N\rightarrow\infty$ by \eqref{eq-Gamma-GammaK-1}.  Finally, $\eqref{g3}\rightarrow 0$ as $N\rightarrow\infty$ follows from proposition \ref{prp-GammaKN-conv-1}.  Thus
\begin{align*}
\lim_{N\rightarrow\infty}\|P_{\le K(N)}\Gamma^{\Phi_N} - \Gamma\|_{L^\infty_{t\in I}\cH_\xi^1}=0.
\end{align*}
This completes the proof of Theorem \ref{derivation}.
\qed


$\;$ \\

\section{Global Well-Posedness}
\label{sec-gwp-1}

In this section, we prove Theorem \ref{gwp}.  
To this end, we first prove positive semidefiniteness of solutions to the GP hieararchy
in  Theorem \ref{pos_again}, below, and 
subsequently global well posedness of the GP hierarchy in Theorem \ref{gwp_again}.

To prove positive semi-definiteness, we recall the quantum de Finetti theorem, Theorem \ref{qdf},
and we invoke the following lemma from \cite{CHPS}.

\begin{lemma}\label{ae}
Let $\mu$ be a Borel probability measure in $L^2(\mathbb{R}^3)$, and assume that
\begin{align}\label{eq-}
     \int d\mu(\phi)\|\phi\|_{H^1}^{2k}\le M^{2k}
\end{align}
holds for some finite constant $M>0$, and all $k\in\mathbb{N}$.  Then,
\begin{align*}
    \mu\big(\big\{\phi\in L^2(\mathbb{R}^3)\big|\|\phi\|_{H^1}>M\big\}\big)=0.
\end{align*}
\end{lemma}

\begin{proof}
From Chebyshev's inequality, we have that
\begin{align*}
\mu\big(\big\{\phi\in L^2(\mathbb{R}^3)\big|\|\phi\|_{H^1}>\lambda\big\}\big)
\le\frac{1}{\lambda^{2k}}\int d\mu(\phi)\|\phi\|_{H^1}^{2k}\le\frac{M^{2k}}{\lambda^{2k}}
\end{align*}
for any $k>0$.  For $\lambda>M$, the right hand side tends to zero when $k\rightarrow\infty$.
\end{proof}

We recall that, for $I\subseteq \mathbb{R}$,
\begin{align*}
\mathcal{W}^\alpha_\xi(I)=\{\Gamma\in C(I,\mathcal{H}^\alpha_\xi)\,|\, B^+\Gamma,B^-\Gamma\in L^2_{loc}(I,\mathcal{H}^\alpha_\xi)\}.
\end{align*}

We are now ready to prove positive semidefiniteness of solutions to the GP hierarchy.

\begin{thm}\label{pos_again}
Assume that 
\eqn 
    \gamma_0^{(k)} = \int d\mu(\phi)(|\phi\rangle\langle\phi|)^{\otimes k}
    \;,\;\;k\in\N\,
\eeqn
satisfies $\Gamma_0=(\gamma_0^{(k)})_{k=1}^\infty\in\mathfrak{H}^1_{\xi'}$ for some $0<\xi'<1$,
where $d\mu$ is a probability measure supported either on the unit sphere, or on the unit ball in $L^2(\R^3)$.
Then, for $0<\xi'<1$ and $\xi>0$ satisfying \eqref{xi_relationship}, and for $0<T<\min\{T_0(\xi),T_1(\xi)\}$ (see \eqref{T_0_def} and \eqref{T_1_def}), there is a unique solution $\Gamma\in\mathcal{W}^1_\xi([0,T])$ to the cubic defocusing GP hierarchy \eqref{gp} in 
$\mathbb{R}^3$ with initial data $\Gamma_0$.  Moreover, $\Gamma(t)$ is positive semidefinite for $t\in [0,T]$.
\end{thm}

\begin{proof}
By \cite{CP} and Proposition \ref{continuous}, there exists a unique solution $\Gamma$ 
to the GP hierarchy in $\mathcal{W}^1_\xi([0,T])$ with initial data $\Gamma_0$.  

By the quantum de Finetti theorem (Theorem \ref{qdf}) and Lemma \ref{ae}, there 
exists a positive semidefinite Borel probability measure $\mu$ on the unit sphere in $L^2(\mathbb{R}^3)$ such that
\begin{align}
    \gamma_0^{(k)}=\int\,d\mu(\phi)(|\phi\ra\la\phi |)^{\otimes k}
\end{align}
and $\|\phi\|_{H^1}^2\le (\xi')^{-1}\|\Gamma_0\|_{\mathfrak{H}^1_{\xi'}}$ $\mu$-almost everywhere.  
Let $S_t$ be the flow map of the cubic defocusing NLS.  Since the NLS is well-posed in $H^1$, 
\begin{align}
    \tilde{\gamma}^{(k)}(t):=\int\,d\mu(\phi)(|S_t\phi\ra\la S_t\phi |)^{\otimes k}
\end{align}
is well-defined, positive semidefinite, and $\tilde{\Gamma}:=\{\tilde{\gamma}^{(k)}\}_{k=1}^\infty$ satisfies the cubic defocusing GP hierarchy.

Moreover, we claim that $\tilde{\Gamma}\in\mathcal{W}^1_\xi([0,T])$.
To prove this fact, let $\la\mathcal{K}^{(m)}\ra_{\Gamma(t)}$, $m\in\N$, denote the higher order energy 
functionals for the cubic GP hierarchy  
introduced in \cite{CPHE}. They are given by
\eqn\label{eq-Km-def-1}
    \langle K^{(m)}\rangle_{\Gamma(t)}:=\text{Tr}_{1,3,5,\dots,2m +1}(K^{(m)}\gamma^{(2m )}(t))\label{HE_def}
\eeqn
for $m\in\mathbb{N}$, where
\begin{align*}
    K_\ell&:=\frac{1}{2}(1-\Delta_{x_\ell})\text{Tr}_{\ell+1}+\frac{1}{4}B^+_{\ell;\ell+1} 
    \;\;\;,\;\;\;\ell\in\N\,, \\
    K^{(m)}&:=K_1K_{3}\cdots K_{2m-1}.
\end{align*}
In \cite{CPHE}, it is shown that these higher order energy functionals are conserved:
\begin{prop}(C-Pavlovi\'c \cite{CPHE})
Suppose that $\Gamma\in\mathfrak{H}^1_\xi$ is symmetric, admissible, and solves the GP hierarchy.  
Then, for all $m\in\mathbb{N}$, the higher order energy functionals \eqref{HE_def} are bounded and conserved, 
$\langle K^{(m)}\rangle_{\Gamma(t)}= \langle K^{(m)}\rangle_{\Gamma(0)}$.
\end{prop}

Using the de Finetti theorem, we can eliminate the requirement of admissibility.
We write 
\eqn
    E[\phi] := \tfrac{1}{2}\|\phi\|_{H^1}^2\|\phi\|_{L^2}^2+\tfrac{1}{4}\|\phi\|_{L^4}^4 = E[S_t\phi]
\eeqn
for the conserved energy of the solution of the NLS. 
Then, it can be easily checked that
\begin{align}
    \la\mathcal{K}^{(m)}\ra_{\tilde\Gamma(t)}
    &= \int d\mu(\phi)\Big(\, \frac12 \, + \, E[S_t\phi] \, \Big)^m \,.
\end{align}
We have that the sequence of higher energy functionals $\la\mathcal{K}^{(m)}\ra_{\Gamma(t)}$, 
for $m\in\N$, satisfies
\begin{align*}
    \|\Gamma(t)\|_{\mathfrak{H}^1_{\xi}}&
    \le\sum_{m\in\mathbb{N}}(2\xi)^m\la\mathcal{K}^{(m)}\ra_{\Gamma(t)}\\
    &=\sum_{m\in\mathbb{N}}(2\xi)^m\la\mathcal{K}^{(m)}\ra_{\Gamma(0)}\\
    &\le\|\Gamma(0)\|_{\mathfrak{H}^1_{\xi'}},
\end{align*}
by Theorem 6.2 in \cite{CPHE}.

As a consequence, we find that
\begin{align*}
    \|\tilde{\Gamma}(t)\|_{\mathcal{H}^1_\xi}
    &\le\|\tilde{\Gamma}(t)\|_{\mathfrak{H}^1_\xi}\\
    &\le\sum_{m\in\mathbb{N}}(2\xi)^m\la\mathcal{K}^{(m)}\ra_{\tilde{\Gamma}(t)}\\
    &=\sum_{m\in\mathbb{N}}(2\xi)^m\la\mathcal{K}^{(m)}\ra_{\tilde{\Gamma}(0)}\\
    &\le\|\Gamma_0\|_{\mathfrak{H}^1_{\xi'}}\\
    &<\infty \,.
\end{align*}
Moreover,
\begin{align}
    &\|B\tilde{\Gamma}\|_{L^2_{t\in [0,T]}H^1_{\xi}}
    \nonumber\\
    &\le\sum_{k=1}^\infty(\xi)^k\int\,d\mu(\phi)\|\la\nabla\ra(|S_t\phi|^2S_t\phi)\|_{L^2_{t\in [0,T]}
    L^2(\mathbb{R}^{3})}
    \|\la\nabla\ra S_t\phi\|_{L^\infty_{t\in [0,T]}L^2(\mathbb{R}^{3})}^{2k-1}
    \label{nls_estimate_1}\\
    &\le\sum_{k=1}^\infty(\xi)^k\int\,d\mu(\phi)\||S_t\phi|^2\|_{L^\infty_tL^3(\mathbb{R}^{3})}
    \|\la\nabla\ra S_t\phi\|_{L^2_{t\in [0,T]}L^6(\mathbb{R}^{3})}
    \|\la\nabla\ra S_t\phi\|_{L^\infty_{t\in [0,T]}L^2(\mathbb{R}^{3})}^{2k-1}
    \nonumber\\
    &\le\sum_{k=1}^\infty(\xi)^k\int\,d\mu(\phi)\|S_t\phi\|^2_{L^\infty_{t\in [0,T]}L^6(\mathbb{R}^{3})}
    \|\la\nabla\ra S_t\phi\|_{L^2_{t\in [0,T]}L^6(\mathbb{R}^{3})}
    \|\la\nabla\ra S_t\phi\|_{L^\infty_{t\in [0,T]}L^2(\mathbb{R}^{3})}^{2k-1} \,.
    \label{nls_estimate_1-1}
\end{align}
Here, we use the bound 
\eqn\label{nls_estimate_2}
    \|\la\nabla\ra S_t\phi\|_{L^2_{t\in [0,T]}L^6(\mathbb{R}^{3})} \leq C(T)\|\la\nabla\ra\phi\|_{L^2}
    \leq C(T) \sqrt{ 1+2E[\phi] }\,,
\eeqn
with $T>0$ as in \eqref{T_1_def}, below;
see for instance \cite{KM} or \cite{CW} for details.
Moreover,
\eqn 
    \|\la\nabla\ra S_t\phi\|_{L^\infty_{t\in [0,T]}L^2(\mathbb{R}^{3})}
    \leq
    \sup_{t\in[0,T]} \sqrt{1+2E[S_t\phi] }
    =
    \sqrt{  1+2E[\phi] }
\eeqn
We then obtain that
\begin{align}
    \eqref{nls_estimate_1-1}&\leq C\sum_{k=1}^\infty(2\xi)^k\int\,d\mu(\phi)
    \Big(\frac12+E[\phi]\Big)^{k+1} \nonumber\\
    &=C\xi^{-1}\sum_{k=2}^\infty(2\xi)^k\la\mathcal{K}^{(k)}\ra_{\tilde{\Gamma}(0)}\nonumber\\
    &\le C\xi^{-1}\|\Gamma_0\|_{\mathfrak{H}^1_{\xi'}}\nonumber\\
    &<\infty.
\end{align}
Finally, we pick $T_1(\xi)>0$ sufficiently small that \eqref{nls_estimate_2} above holds for
\begin{align}
     \label{T_1_def}
    0<T<T_1(\xi)\,,
\end{align} 
noting that the constant $C(T)$ in \eqref{nls_estimate_2}  depends on 
$\|\phi\|_{H^1}<(\xi')^{-1/2}\|\Gamma_0\|_{\frH_{\xi'}^1}$ and thus on $\xi$, 
where $\xi$ and $\xi'$ are related as in
\eqref{xi_relationship}.   

Thus, we have shown that $\tilde{\Gamma}\in\mathcal{W}^1_\xi([0,T])$.
By uniqueness of solutions to the GP hierarchy in $\mathcal{W}^1_\xi([0,T])$, we 
conclude that $\Gamma=\tilde{\Gamma}$.  

In particular, we note that ${\Gamma}(t)$ is positive semidefinite for $t\in [0,T]$.
\end{proof}

Now that we have positive semidefinitenss of solutions to the GP hierarchy, we are able to to global well posedness of solutions to the GP hierarchy, using an induction argument as in \cite{CPHE} below.

\begin{thm}\label{gwp_again}
Suppose that $\Gamma_0=(\gamma_0^{(k)})_{k=1}^\infty\in\mathfrak{H}^1_{\xi'}$ is as in Theorem \ref{pos_again}.  
Then, for $0<\xi'<1$ and $\xi_1$ satisfying \eqref{xi_relationship}, there is a unique global solution 
$\Gamma\in\mathcal{W}^1_{\xi_1}(\mathbb{R})$ to the cubic defocusing GP hierarchy \eqref{gp} in $\mathbb{R}^3$ 
with initial data $\Gamma_0$.  Moreover, $\Gamma(t)$ is positive semidefinite and satisfies
\begin{align}
    \|\Gamma(t)\|_{\mathcal{H}_{\xi_1}^1}\le \|\Gamma_0\|_{\mathfrak{H}_{\xi'}^1}
\end{align}
for all $t\in\mathbb{R}$.
\end{thm}

\begin{proof} 
Let $I_j$ be the time interval $[jT,(j+1)T]$, where $0<T<\min\{T_0(\xi_1),T_1(\xi_1)\}$ (see \eqref{T_0_def} and \eqref{T_1_def}) and $\xi,\xi_1$ satisfy \eqref{xi_relationship}.  By \cite{CP} and Proposition \ref{continuous}, we have that there is a unique solution $\Gamma$ to the GP hierarchy in $\mathcal{W}^1_\xi(I_0)$.  Moreover, by Theorem \ref{pos_again}, $\Gamma$ is positive semidefinite on $I_0$.  It follows as in the proof of Theorem 7.2 in \cite{CPHE} that the higher order energy functionals $\langle K^{(m)}\rangle_{\Gamma(t)}$, which are defined in equation \eqref{HE_def}, are conserved on $I_0$.  Thus, as in inequality (7.18) in \cite{CPHE}, we have that on $I_0$,
\begin{align}
\|\Gamma(t)\|_{\mathcal{H}^1_\xi}
&\le\|\Gamma(t)\|_{\mathfrak{H}^1_\xi}\label{energy1}\\
&\le\sum_{m\in\mathbb{N}}(2\xi)^m\langle K^{(m)}\rangle_{\Gamma(t)}\label{energy2}\\
&\le\sum_{m\in\mathbb{N}}(2\xi)^m\langle K^{(m)}\rangle_{\Gamma_0}\label{energy3}\\
&\le\|\Gamma_0\|_{\mathfrak{H}^1_{\xi'}.}\label{energy4}
\end{align}
Note that positive semidefiniteness of $\Gamma$ is needed to pass from \eqref{energy1} to \eqref{energy2} because the definition of $\|\Gamma(t)\|_{\mathfrak{H}^1_\xi}$ involves taking absolute values, but the definition of $\langle K^{(m)}\rangle_{\Gamma(t)}$ does not.  

Therefore $\Gamma(T)\in\mathfrak{H}^1_\xi$, and so by \cite{CP} and Proposition \ref{continuous}, there is a unique solution $\Gamma\in \mathcal{W}^1_{\xi_1}(I_1)$ of the GP hierarchy with initial data $\Gamma(T)$.  By another application of Theorem \ref{pos_again} and energy conservation \eqref{energy1} $\sim$ \eqref{energy4}, $\Gamma$ is positive semidefinite on $I_1$ and $\Gamma(2T)\in\mathfrak{H}^1_\xi$.  Thus, we can repeat the argument and find that we have a unique solution $\Gamma\in\mathcal{W}^1_{\xi_1}(\mathbb{R})$.  Moreover,
\begin{align*}
\|\Gamma(t)\|_{\mathcal{H}_{\xi_1}^1}\le
\|\Gamma(t)\|_{\mathcal{H}_\xi^1}\le \|\Gamma_0\|_{\mathfrak{H}_{\xi'}^1}
\end{align*}
for all $t\in\mathbb{R}$.
\end{proof}

$\;$\\


\section{Global derivation of the GP hierarchy}\label{global derivation}

In this section, we show that the validity of Theorem \ref{derivation} can be extended to arbitrarily large values of $T$, provided that $\Gamma_0\in\mathfrak{H}^1_{\xi'}$ has the form \eqref{eq-gamma0-deF-1}, and that $\xi$ is sufficiently small.
This is obtained from combining Theorem \ref{derivation} and Theorem \ref{gwp}
in a recursive manner.

We begin by observing that, in the statement of Theorem \ref{derivation}, instead of assuming that
\begin{align*}
\Gamma_0:=\lim_{N\rightarrow\infty} \Gamma^{\Phi_N}(0)
\end{align*}
holds in $\mathcal{H}^1_{\xi'}$, we may assume that
\begin{align*}
\Gamma_0:=\lim_{N\rightarrow\infty} P_{\le K(N)}\Gamma^{\Phi_N}(0)
\end{align*}
holds.  Indeed, the proof of Theorem \ref{derivation} is unaffected by this replacement.

We also note that initial condition $\la \Phi_N(0),H_N^k \Phi_N(0)\ra$ implies that $\Gamma^{\Phi_N}(t)\in\mathcal{H}^1_{\xi'}$ for any $t\in\mathbb{R}$, provided that $\xi'<(4(C+1))^{-1}$.  This follows from \eqref{bound3b}.  In fact, given $\xi'$, we have a bound $\widetilde C$, uniform in $N$ and $t$, such that
\begin{align}
\|\Gamma^{\Phi_N}(t)^{(k)}\|_{H^1}<\widetilde C^k.\label{ubound}
\end{align}

We also note that, by Theorem \ref{gwp}, the solution to the GP hierarchy $\Gamma(t)\in\mathcal{H}^1_{\xi_1}$ for all $t\in\mathbb{R}$, provided that $\xi_1$ sufficiently small.

Thus, under the assumptions of Theorem \ref{derivation}, at time $T$, we have
\begin{align}
\begin{cases}
&\Gamma^{\Phi_N}(T)\in\mathcal{H}^1_{\xi_1}\text{ and}\\
&\Gamma(T)=\lim_{N\rightarrow\infty} P_{\le K(N)}\Gamma^{\Phi_N}(T)\text{ in }\mathcal{H}^1_{\xi_1},
\end{cases}\label{T}
\end{align}
provided that $\xi_1$ is sufficiently small (note that we also require $\xi_1<(4(C+1))^{-1}$).
By another application of Theorem \ref{derivation}, we have that at time $2T$,
\begin{align}
&\Gamma^{\Phi_N}(2T)\in\mathcal{H}^1_{\xi_1}\text{ and}\nonumber\\
&\Gamma(2T)=\lim_{N\rightarrow\infty} P_{\le K(N)}\Gamma^{\Phi_N}(2T)\text{ in }\mathcal{H}^1_{\xi_2},\label{xi2 bound}
\end{align}
provided that $\xi_2<\xi_1$ is sufficiently small.  \eqref{xi2 bound} says that
\begin{align*}
\sum_{k=1}^\infty \xi_2^k \|\Gamma(2T)^{(k)}-P_{\le K(N)}\Gamma^{\Phi_N}(2T)^{(k)}\|_{H^1}\rightarrow 0\text{ as }N\rightarrow\infty.
\end{align*}
However, by \eqref{ubound} and the dominated convergence theorem for sequences, we actually have the stronger statement 
\begin{align*}
\sum_{k=1}^\infty \xi_1^k \|\Gamma(2T)^{(k)}-P_{\le K(N)}\Gamma^{\Phi_N}(2T)^{(k)}\|_{H^1}\rightarrow 0\text{ as }N\rightarrow\infty,
\end{align*}
where we have $\xi_1$ instead of $\xi_2$.  Thus, at time $2T$ we actually have
\begin{align}
\begin{cases}
&\Gamma^{\Phi_N}(2T)\in\mathcal{H}^1_{\xi_1}\text{ and}\\
&\Gamma(2T)=\lim_{N\rightarrow\infty} P_{\le K(N)}\Gamma^{\Phi_N}(2T)\text{ in }\mathcal{H}^1_{\xi_1}.
\end{cases}\label{2T}
\end{align}
Note that \eqref{2T} is the same as \eqref{T}, but with $T$ replaced by $2T$.  Thus, we may iterate the argument again, and conclude that Theorem \ref{derivation} holds for $T$ arbitrarily large,  provided that $\Gamma_0\in\mathfrak{H}^1_{\xi'}$ has the form \eqref{eq-gamma0-deF-1},  and that $\xi$ is sufficiently small.

$\;$\\

\appendix

\section{Strong vs weak-* convergence}
\begin{prop}\label {strong weak}
Suppose that $(\gamma_N^{(k)})_{N=1}^\infty$ is a sequence of operators on $L^2(\mathbb{R}^{k})$ such that $\gamma_N^{(k)}\rightarrow\gamma_\infty^{(k)}$ strongly in Hilbert Schmidt norm.  Suppose also that $\gamma_N^{(k)}$ and $\gamma^{(k)}_\infty$ are trace class operators such ${\rm Tr}
|\gamma_N^{(k)}|\le 1$ for all $N$.  Then $\gamma_N^{(k)}\rightarrow\gamma_\infty^{(k)}$ in the weak-* topology induced by the trace norm.
\begin{proof}
We follow the usual construction of a metric for the weak-* topology induced by the trace norm, as presented in \cite{esy2}, for example.  Let $\mathcal{K}_k$ be the space of compact operators on $L^2(\mathbb{R}^{k})$ equipped with the operator norm topology.  Let $\mathcal{L}^1_k$ be the space of trace class operators on $L^2(\mathbb{R}^{2k})$.  By \cite{rs}, we have that $\mathcal{L}^1_k=\mathcal{K}_k^*$.  Since $\mathcal{K}_k$ is separable, there exists a sequence $\{J_i^{(k)}\}_{i=1}^\infty\in\mathcal{K}_k$ of Hilbert Schmidt operators, dense in the unit ball of $\mathcal{K}_k$.  Note that Hilbert Schmidt operators are dense in the space of compact operators, because, by \cite{rs}, every compact operator on a Hilbert space is of the form $\lim_{N\rightarrow\infty}\sum_{n=1}^N \lambda_n\la\psi_n,\cdot\,\ra\phi_n$, with $\{\psi_n\}_{n=1}^\infty$ and $\{\phi_n\}_{n=1}^\infty$ orthonormal sets, and $\{\lambda_n\}_{n=1}^\infty$ positive real numbers such that $\lambda_n\rightarrow 0$.  On $\mathcal{L}^1_k$, we define the metric $\eta_k$ by
\begin{align*}
\eta_k(\gamma^{(k)},\widetilde\gamma^{(k)}):=\sum_{i=1}^\infty 2^{-i}\bigg|{\rm Tr} J_i^{(k)}\big(\gamma^{(k)}-\widetilde\gamma^{(k)}\big)\bigg|.
\end{align*}
By \cite{rudin}, the topology induced by the metric $\eta_k$ is equivalent to the weak-* toplology on $\mathcal{L}^1_k$.

Now, since $\{J_i^{(k)}\}_{i=1}^\infty\in\mathcal{K}_k$ are Hilbert Schmidt, we have
\begin{align}
    {\rm Tr}|J_i^{(k)}(\gamma_N^{(k)}-\gamma_\infty^{(k)})|&\le ({\rm Tr}(|J_i^{(k)}|^2))^{1/2}
    ({\rm Tr}(|\gamma_N^{(k)}-\gamma_\infty^{(k)}|^2))^{1/2}
    \\&\rightarrow 0\text{ as }N\rightarrow\infty.
\end{align}
Moreover,
\begin{align}
    {\rm Tr}|J_i^{(k)}(\gamma_N^{(k)}-\gamma_\infty^{(k)})|
    &\le \|J_i^{(k)}\|_{L^2\rightarrow L^2} \,\,{\rm Tr}|\gamma_N^{(k)}-\gamma_\infty^{(k)}|\\
    &\le 1+{\rm Tr}|\gamma_\infty^{(k)}|.
\end{align}
Thus, by the dominated convergence theorem for sequences, $\eta_k(\gamma_N^{(k)},\gamma_\infty^{(k)})\rightarrow 0$ as $N\rightarrow\infty$, and so $\gamma_N^{(k)}\rightarrow\gamma_\infty^{(k)}$ in the weak-* topology on $\mathcal{L}^1_k$.
\end{proof}
\end{prop}

\newpage

\section{Conservation of admissibility for the GP hierarchy}
In this part of the appendix, we prove that the GP hierarchy conserves admisibility.  This result has been used in many papers, but we have not found an explicit proof.  For the convenience of the reader, we present it here.

\begin{prop}\label{adm_prop}
Suppose that $\Gamma_0=(\gamma_0^{(k)})_{k=1}^\infty\in\mathcal{H}^1_{\xi'}$ is admissible and satisfies $\rm{Tr}\,\gamma_0^{(k)}=1$ for all $k\in\mathbb{N}$.  Then, for $0<\xi'<1$ and $\xi$ satisfying \eqref{xi_relationship}, the unique solution $\Gamma\in\mathcal{W}^1_\xi(I)$ to the GP hierarchy obtained in \cite{CP} is admissible for all $t\in I$, provided that $A:=\{A^{(k)}\}_{k=1}^\infty\in\mathcal{W}^1_\xi(I)$, where
\begin{align}
A^{(k)}(t,\underline{x}_k;\underline{x}_k'):=
-\gamma^{(k)}(t,\underline{x}_k;\underline{x}_k')+
\int \gamma^{(k+1)}(t,\underline{x}_k,x_{k+1};\underline{x}_k',x_{k+1})\,dx_{k+1}.\label{a}
\end{align}
\end{prop}

\begin{proof} 
We first note that for $f\in \mathcal{S}(\mathbb{R}^n\times\mathbb{R}^n)$, we have
\begin{equation}\label{lma}
\int((\Delta_{x_1}-\Delta_{x_2})f)(x,x)\,dx=0.
\end{equation} 
Indeed, this follows from
\begin{align*}
&\int((\Delta_{x_1}-\Delta_{x_2})f)(x,x)\,dx\\
&=\int\int\delta(x_1-x_2)(\Delta_{x_1}-\Delta_{x_2})f(x_1,x_2)\,dx_1\,dx_2\\
&=\int\int\int\int\delta(x_1-x_2)e^{iu_1x_1+iu_2x_2}((u_2)^2\hat{f}(u_1,u_2)-u_1^2\hat{f}(u_1,u_2))\,du_1\,du_2\,dx_1\,dx_2\\
&=\int\int\int e^{ix_1(u_1+u_2)}((u_2)^2\hat{f}(u_1,u_2)-u_1^2\hat{f}(u_1,u_2))\,du_1\,du_2\,dx_1\\
&=\int\int \delta(u_1+u_2)((u_2)^2\hat{f}(u_1,u_2)-u_1^2\hat{f}(u_1,u_2))\,du_1\,du_2\\
&=\int (u_1^2-u_1^2)\hat{f}(u_1,-u_1)\,du_1\\
&=0,
\end{align*}
which implies \eqref{lma}.

Next, we note that the definition of admissibility implies that $\gamma^{(k)}$ is admissible at time $t$ if and only if
\begin{align*}
A^{(k)}(t,\underline{x}_k,\underline{x}_k')=0.
\end{align*}
Since $\Gamma$ satisfies the GP hierarchy, we have that
\begin{align}
i\partial_t A^{(1)}&(x_1;x_1')
\nonumber\\
&=(\Delta_{x_1}-\Delta_{x_1'})\gamma^{(1)}(x_1;x_1')\label{a11}\\
&\hspace{1cm} -\kappa_0\bigg[\gamma^{(2)}(x_1,x_1;x_1',x_1)-\gamma^{(2)}(x_1,x_1';x_1',x_1')\bigg]\label{a21}\\
&\hspace{1cm} +\int\bigg[ \left((-\Delta_{\underline{x}_2}+\Delta_{\underline{x}_2'})\gamma^{(2)}\right)(x_1,x_2;x_1',x_2)\label{L1}\\
&\hspace{2cm} +\kappa_0\gamma^{(3)}(x_1,x_2,x_1;x_1',x_2,x_1)\nonumber\\
&\hspace{2cm} -\kappa_0\gamma^{(3)}(x_1,x_2,x_1';x_1',x_2,x_1')\nonumber\\
&\hspace{2cm} +\kappa_0\gamma^{(3)}(x_1,x_2,x_2;x_1',x_2,x_2)\nonumber\\
&\hspace{2cm} -\kappa_0\gamma^{(3)}(x_1,x_2,x_2;x_1',x_2,x_2)\bigg]\,dx_2\nonumber\\
&=\int(\Delta_{x_1}-\Delta_{x_1'})\gamma^{(2)}(x_1,x_2;x_1',x_2)\,dx_2\label{a12}\\
&\hspace{1cm}-(\Delta_{x_1}-\Delta_{x_1'})A^{(1)}(x_1;x_1')\nonumber\\
&\hspace{1cm} -\kappa_0\int\bigg[\gamma^{(3)}(x_1,x_1,x_2;x_1',x_1,x_2)-\gamma^{(3)}(x_1,x_1',x_2;x_1',x_1',x_2)\bigg]\,dx_2\label{a22}\\
&\hspace{1cm}+\kappa_0A^{(2)}(x_1,x_1;x_1',x_1)-\kappa_0A^{(2)}(x_1,x_1';x_1',x_1')\nonumber\\
&\hspace{1cm} +\int\bigg[ \left((-\Delta_{x_1}+\Delta_{x_1'})\gamma^{(2)}\right)(x_1,x_2;x_1',x_2)\label{L2}\\
&\hspace{2cm} +\kappa_0\gamma^{(3)}(x_1,x_2,x_1;x_1',x_2,x_1)\nonumber\\
&\hspace{2cm} -\kappa_0\gamma^{(3)}(x_1,x_2,x_1';x_1',x_2,x_1')\bigg]\,dx_2\nonumber\\
&=-(\Delta_{x_1}-\Delta_{x_1'})A^{(1)}(x_1;x_1'),\label{sym2}\\%
&\hspace{1cm}+\kappa_0A^{(2)}(x_1,x_1;x_1',x_1)-\kappa_0A^{(2)}(x_1,x_1';x_1',x_1')\nonumber
\end{align}
where \eqref{a} was used to pass from \eqref{a11} to \eqref{a12} and from \eqref{a21} to \eqref {a22}. 
Moreover, \eqref{lma} and density of $\mathcal{S}$ in $H^1$ was used to pass from \eqref{L1} to \eqref{L2}.  Symmetry of $\gamma^{(k)}$ was used to pass to \eqref{sym2}.

Observe that \eqref{sym2} is precisely the right hand side of the first equation in the GP hierarchy.  Thus $A^{(1)}$, and similarly $A^{(k)}$ for $k>1$, satisfies the GP hierarchy.  $A(0)=0$, so by uniquenss of solutions to the GP hierarchy \cite{CP}, $A=0$.
\end{proof}

$\;$\\

\section{Continuity of solutions to the GP hierarchy}
In \cite{CP}, it is shown that there is a unique soution $\Gamma$ to the GP hierarchy \eqref{gp} in $\{\Gamma\in L^\infty_{t\in [0,T]}\mathcal{H}^\alpha_\xi\,|\, B^+\Gamma,B^-\Gamma\in L^2_{t\in [0,T]}\mathcal{H}^\alpha_\xi\}$.  In this part of the appendix, we show that this solution $\Gamma$ is an element of $C([0,T],\mathcal{H}^1_\xi)$.
\begin{lemma}\label{free1}
If $\gamma^{(k)}\in L^2(\mathbb{R}^{dk}\times\mathbb{R}^{dk})$, then
\begin{align*}
\lim_{t\rightarrow 0}\|\left(U^{(k)}(t)-U^{(k)}(0)\right)\gamma^{(k)}\|_{L^2(\mathbb{R}^{dk}\times\mathbb{R}^{dk})}=0.
\end{align*}
\end{lemma}
\begin{proof}
We recall that $U^{(k)}(t)=e^{-it(-\Delta_{\underline{x}_k}+\Delta_{\underline{x}_k'})}$.  Since $-\Delta_{\underline{x}_k}+\Delta_{\underline{x}_k'}$ is a self-adjoint operator, we have from theorem VIII.7 in \cite{rs} that $U^{(k)}(t)$ is a strongly continuous one-parameter unitary group, and the lemma follows.
\end{proof}

\begin{lemma}\label{free2}
If $\Gamma\in \mathcal{H}^\alpha_\xi(\mathbb{R}^{dk}\times\mathbb{R}^{dk})$, then
\begin{align*}
\lim_{t\rightarrow 0}\|\left(U(t)-U(0)\right)\Gamma\|_{\mathcal{H}^\alpha_\xi}=0.
\end{align*}
\end{lemma}
\begin{proof}
\begin{align*}
&\|\left(U(t)-U(0)\right)\Gamma\|_{\mathcal{H}^\alpha_\xi}\\
&=\sum_{k=1}^\infty\xi^k\|\left(U^{(k)}(t)-U^{(k)}(0)\right)S^{(k,\alpha)}\gamma^{(k)}\|_{L^2}\\
&\rightarrow 0\text{ as }t\rightarrow 0
\end{align*}
by Lemma \ref{free1}, the fact that $\|U^{(k)}(t)\|_{L^2\rightarrow L^2}\le 1$, and the dominated convergence theorem for series.
\end{proof}
\begin{prop}\label{continuous}
The solution $\Gamma$ to the GP hierarchy constructed in \cite{CP} lies in $C([0,T],\mathcal{H}^1_\xi)$.
\end{prop}

\begin{proof}
As proven in \cite{CP}, the solution $\Gamma$ satisfies
\begin{align}
\Gamma&\in L^\infty_{t\in [0,T]}\mathcal{H}^1_\xi,\label{gamma_space}\\
B\Gamma&\in L^2_{t\in [0,T]}\mathcal{H}^1_\xi,\text{ and}\label{bgamma_space}\\
\Gamma(t)&=U(t)\Gamma_0+i\kappa_0\int_0^tU(t-s)B\Gamma(s)\,ds.\label{gamma_equation}
\end{align}
Thus, in $\mathcal{H}^1_\xi$, we have that
\begin{align}
&\lim_{h\rightarrow 0}\bigg[\Gamma(t+h)-\Gamma(t)\bigg]\nonumber\\
&=\lim_{h\rightarrow 0}\bigg[U(t+h)\Gamma_0+i\kappa_0\int_0^{t+h}U(t+h-s)B\Gamma(s)\,ds\nonumber\\
&\hspace{1cm}-U(t)\Gamma_0-i\kappa_0\int_0^tU(t-s)B\Gamma(s)\,ds\bigg]\nonumber\\
&=\lim_{h\rightarrow 0}U(h)\Gamma_0\label{term1}\\
&\hspace{1cm}+\lim_{h\rightarrow 0}(U(h)-U(0))\underbrace{i\kappa_0\int_0^tU(t-s)B\Gamma(s)\,ds}_{[*]}\label{term2}\\
&\hspace{1cm}+\lim_{h\rightarrow 0}i\kappa_0\int_t^{t+h}U(t+h-s)B\Gamma(s)\,ds\label{term3}.
\end{align}
By Lemma \ref{free2}, $\eqref{term1}=0$.  By \eqref{gamma_space} and \eqref{gamma_equation}, $[*]\in\mathcal{H}^1_\xi$, so it follows from Lemma \ref{free2} that $\eqref{term2}=0$.  Now
\begin{align*}
&\left\|\int_t^{t+h}U(t+h-s)B\Gamma(s)\,ds\right\|_{\mathcal{H}^1_\xi}\\
&\le\int_t^{t+h}\left\|U(t+h-s)B\Gamma(s)\right\|_{\mathcal{H}^1_\xi}\,ds\\
&\le\sqrt{h}\|U(t+h-s)B\Gamma(s)\|_{L^2_{s\in [t,t+h]}\mathcal{H}^1_\xi}\\
&=\sqrt{h}\|B\Gamma(s)\|_{L^2_{s\in [t,t+h]}\mathcal{H}^1_\xi}\\
&\rightarrow 0\text{ as }h\rightarrow 0
\end{align*}
by \eqref{bgamma_space}, so $\eqref{term3}=0$.
\end{proof}

\section{Iterated Duhamel formula and boardgame argument}

In this part of the appendix, we recall a technical result from \cite{CPBBGKY}
that is used in parts of this paper. It corresponds to Lemma B.3 in \cite{CPBBGKY}. 
 
Let $\Xi=({\Xi^{(k)}})_{n\in\N}$ denote 
a sequence of functions 
${\Xi^{(k)}}\in L^2_{t\in[0,T]}H^1(\R^{3k}\times\R^{3k})$, for  $T>0$.
Then, we define the associated sequence $\duh_j(\Xi)$ 
of  {\em $j$-th level iterated Duhamel terms} based on $B_N^{main}$
(see Section \ref{ssec-BBGKY-1} for notations), with components given by
\eqn\label{eq-Duh-j-def-1}
	\lefteqn{
	\duh_j({\Xi})^{(k)}(t) 
	}
	\\
	& := & i^j\int_0^t dt_1 \cdots \int_0^{t_{j-1}}dt_j
	B_{N;k+1}^{main}
	e^{i(t-t_1)\Delta_\pm^{(k+1)}}
	B_{N;k+2}^{main}e^{i(t_1-t_2)\Delta_\pm^{(k+2)}}
	\nonumber\\
	&&\quad\quad\quad\quad\quad\quad
	B_{N;k+2}^{main}\cdots 
	\cdots e^{i (t_{j-1}-t_j) \Delta_\pm^{(k+j)}}   
	( \, {\Xi} \, )^{(k+j)}(t_j) \,, \;\;
	\nonumber
\eeqn 
with the conventions $t_0:=t$, and 
\eqn
	\duh_0(\Xi)^{(k)}(t) \, := \, ( \, \Xi \, )^{(k)}(t)
\eeqn
for $j=0$.  Using the boardgame estimates of \cite{esy1,esy2,KM}, one obtains:

\begin{lemma}
\label{lm-boardgame-est-1} 
For $\Xi=({\Xi^{(k)}})_{k\in\N}$ as above,
\eqn\label{eq-BGamma-Duhj-combin-bd-1}
	\lefteqn{
	\| \, \duh_j(\Xi)^{(k)}(t) 
	\, \|_{L^2_{t\in I}H^1(\R^{3k}\times\R^{3k})} 
	}
	\\
	&&\hspace{1cm}
	\, \leq \, k \, C_0^k \, (c_0 T)^{\frac {j}2} \|{\Xi^{(k+j)}}
	\|_{L^2_{t\in I}H^1(\R^{3(k+j)}\times\R^{3(k+j)})}  \,,
	\nonumber
\eeqn 
where the constants $c_0,C_0$ depend only on $d,p$.  For this work, the dimension is given by $d=3$ and the nonlinearity is given by $p=2$ (cubic GP hierarchy).
\end{lemma}

Lemma \ref{lm-boardgame-est-1} is used for the proof of the next result (by suitably exploiting the
splitting $\opB_N=\opB_N^{main}+\opB_N^{error}$), which 
corresponds to Lemma B.3 in \cite{CPBBGKY}.

\begin{lemma}
\label{lm-BGamma-Cauchy-1}
Let $\delta'>0$ be defined by
\eqn
	\beta & = & \frac{1-\delta'}{4}  \, .
\eeqn
Assume that $N$ is sufficiently large that the condition
\eqn
	K & < & \frac{\delta'}{\log C_0} \, \log N \,,
\eeqn
holds, where the constant $C_0$ is as in Lemma \ref{lm-boardgame-est-1}.

Assume that  
$\Xi_N^K\in L^2_{t\in I}\cH_{\xi'}^1$ for some $0<\xi'<1$, 
and that $\xi$ is small enough that $0<\xi<\eta\xi'$,  with
\eqn\label{eq-eta-ineq-def-1}
	\eta \, < \, (\max\{1,C_0\})^{-1} \,.
\eeqn
Let $\Theta_N^K$ and $\Xi_N^K$ satisfy the integral equation
\eqn\label{eq-tildTheta-eq-1}
	\Theta_N^K(t) \, = \, \Xi_N^K(t) 
	\, + \, i \int_0^t \opB_N \, U(t-s) \, \Theta_N^K(s) ds 
\eeqn
The superscript "$K$"  in  $\Theta_N^K$ and $\Xi_N^K$ means that only the first $K$ components  are nonzero,
and $\opB_N=\opB_N^{main}+\opB_N^{error}$.

Then, the estimate
\eqn 
	\| \Theta_N^K\|_{L^2_{t\in I}\cH_\xi^1} 
	\, \leq \, C_1(T,\xi,\xi') \,  \|\Xi_N^K\|_{L^2_{t\in I}\cH_{\xi'}^1}
	\label{eq-Gammadiff-Cauchy-aux-2}
\eeqn
holds for a finite constant $C_1(T,\xi,\xi')>0$ independent of $K,N$. 
\end{lemma}

$\;$ 

\subsection*{Acknowledgements} 
The work of T.C. was supported by NSF grants DMS-1009448
and DMS-1151414 (CAREER).

\end{document}